\documentclass[11pt]{article}

\usepackage[latin1]{inputenc}
\usepackage{fancyhdr}
\usepackage{indentfirst}
\usepackage{graphicx}
\usepackage{newlfont}

\usepackage{epsfig}
\usepackage{amssymb}
\usepackage{amsmath}
\usepackage{latexsym}
\usepackage{amsthm}
\usepackage{bbold}

\def\var{\varepsilon}
\theoremstyle{plain}
\newtheorem{theorem}{Theorem}[section]
\newtheorem{proposition}[theorem]{Proposition}
\newtheorem{lemma}[theorem]{Lemma}
\newtheorem{corollary}[theorem]{Corollary}

\newtheorem{remark}[theorem]{Remark}
\newtheorem{definition}[theorem]{Definition}

\newcommand{\R}{\mathbb{R}}
\newcommand{\N}{\mathbb{N}}

\def\pa{\partial}

\begin{document}

\title{Homogenization of the discrete diffusive coagulation-fragmentation 
equations in perforated domains.}

\author{Laurent Desvillettes$^1$ and Silvia Lorenzani$^2$}

\date{\em $^{1}$ Universit\'{e} Paris Diderot, Sorbonne Paris Cit\'{e}, Institut de Math\'{e}matiques de Jussieu-Paris Rive Gauche, UMR 7586, CNRS, Sorbonne Universit\'{e}s, UPMC Univ. Paris 06, F-75013, Paris, France
\\
$^{2}$ Dipartimento di Matematica, Politecnico di Milano, Piazza Leonardo da
Vinci 32, 20133 Milano, Italy}

\maketitle

\begin{abstract}

The asymptotic behavior of the solution of an infinite set of Smoluchowski's
discrete coagulation-fragmentation-diffusion equations with non-homogeneous
Neumann boundary conditions, defined in a periodically perforated domain,
is analyzed.
Our homogenization result, based on Nguetseng-Allaire two-scale convergence,
is meant to pass from a microscopic model (where the physical processes are
properly described) to a macroscopic one (which takes into account only the
effective or averaged properties of the system).
When the characteristic size of the perforations vanishes, the information
given on the microscale by the non-homogeneous Neumann boundary condition
is transferred into a global source term appearing in the limiting 
(homogenized) equations.
Furthermore, on the macroscale, the geometric structure of the perforated
domain induces a correction in the diffusion coefficients.

\end{abstract}

%\maketitle

\section{Introduction} \label{section0}

This paper is devoted to the homogenization of an infinite set of
Smoluchowski's discrete coagulation-fragmentation-diffusion equations
in a periodically perforated domain.
The system of evolution equations considered describes the dynamics of cluster
growth, that is the mechanisms allowing clusters to coalesce to form larger
clusters or break apart into smaller ones.
These clusters can diffuse in space with a diffusion constant which depends
on their size.
Since the size of clusters is not limited {\it{a priori}}, the system of
reaction-diffusion equations that we consider consists of an infinite number
of equations.
The structure of the chosen equations, defined in a perforated medium with
a non-homogeneous Neumann condition on the boundary of the perforations,
is useful for investigating several phenomena arising in porous media
\cite{11b}, \cite{9a}, \cite{11c} or in the field of biomedical research 
\cite{10}.

Typically, in a porous medium, the domain consists of two parts:
a fluid phase where colloidal species or chemical substances, transported
by diffusion, are dissolved and a solid skeleton (formed by grains or
pores) on the boundary of which deposition processes or chemical reactions
take place.
In recent years, the Smoluchowski equation has been also considered in
biomedical research to model the aggregation and diffusion of
$\beta$-amyloid peptide (A$\beta$) in the cerebral tissue, a process thought to be
associated with the development of Alzheimer's disease.
One can define a perforated geometry, obtained by removing from a fixed
domain (which represents the cerebral tissue) infinitely many small holes
(the neurons).
The production of A$\beta$ in monomeric form from the neuron membranes can be
modeled by coupling the Smoluchowski equation for the concentration of
monomers with a non-homogeneous Neumann condition on the boundaries of
the holes.

%The mathematical complexity underlying the models that can be proposed to
%describe such processes has been fully addressed in our work.
%Furthermore, 
The results of this paper constitute a generalization of some
of the results contained in \cite{11b}, \cite{10}, by considering an infinite
system of equations where both the coagulation and fragmentation processes
are taken into account.
Unlike previous theoretical works, where existence and uniqueness of solutions
for an infinite system of coagulation-fragmentation equations (with
homogeneous Neumann boundary conditions) have been studied 
\cite{22}, \cite{12}, 
we focus in this paper on a distinct aspect, that is, the averaging of the system
of Smoluchowski's equations over arrays of periodically-distributed
microstructures.

Our homogenization result, based on Nguetseng-Allaire two-scale convergence
\cite{16}, \cite{3},
is meant to pass from a microscopic model (where the physical processes
are properly described) to a macroscopic one (which takes into account only
the effective or averaged properties of the system).

\subsection{Setting of the problem} \label{subsec01}

Let $\Omega$ be a bounded open set in $\mathbb {R}^3$ with a smooth boundary
$\partial \Omega$.
Let $Y$ be the unit periodicity cell $[0,1[^3$ (having the paving property).
We perforate $\Omega$ by removing from it a set $T_{\epsilon}$ of
periodically distributed holes defined as follows.
Let us denote by $T$ an open subset of $Y$ with a smooth boundary $\Gamma$,
such that $\overline {T} \subset \mathrm{Int}\, Y$.
Set $Y^{\ast}=Y \setminus T$ which is called in the literature the solid or 
material part.
We define $\tau (\epsilon \overline {T})$ to be the set of all translated
images of $\epsilon \overline {T}$ of the form $\epsilon (k+\overline {T})$,
$k \in \mathbb {Z}^3$.
Then,
$$T_{\epsilon}:=\Omega \cap \tau (\epsilon \overline {T}).$$
Introduce now the periodically perforated domain $\Omega_{\epsilon}$ defined
by
$$\Omega_{\epsilon}=\Omega \setminus \overline {T}_{\epsilon}.$$

For the sake of simplicity, we make the following standard assumption on the 
holes \cite{8a}, \cite{9b}:  there exists a 'security' zone around $\partial \Omega$ without 
holes, that is the holes do not intersect the boundary $\pa\Omega$, so that
%\begin{equation}\label{security}
%\exists \; \delta >0 \;  \mbox{such that}
%\; dist \, (\partial \Omega, T_{\epsilon}) \geq \delta.
%\end{equation}
%Therefore,
 $\Omega_{\epsilon}$ is a connected set.
%\bigskip

%*** Above, does $\delta$ depend on $\epsilon$? If yes, I am not sure that such a set $\Omega$ exists. Laurent ***
\medskip

The boundary $\partial \Omega_{\epsilon}$ of $\Omega_{\epsilon}$ is then 
composed of two parts.
The first one is the union of the boundaries of the holes strictly contained
in $\Omega$.
It is denoted by $\Gamma_{\epsilon}$ and is defined by 
$$\Gamma_{\epsilon}:=\cup \bigg\{  \partial (\epsilon (k+\overline {T})) \mid
\epsilon (k+\overline {T}) \subset \Omega  \bigg \}.$$
The second part of $\partial \Omega_{\epsilon}$ is its fixed exterior
boundary denoted by $\partial \Omega$.
It is easily seen that (see \cite{3a}, Eq. (3))

\begin{equation}\label{a}
\lim_{\epsilon \rightarrow 0} \epsilon \mid \Gamma_{\epsilon} {\mid}_{2}=
\mid \Gamma {\mid}_{2} \frac{\displaystyle \mid \Omega \mid_3}
{\displaystyle \mid Y \mid_3},
\end{equation}
where $\mid \cdot \mid_N$ is the $N$-dimensional Hausdorff measure.
\medskip

The previous definitions and Assumptions on $\Omega$ (and, $T$, $\Gamma$, $\Omega_{\epsilon}$, $T_{\epsilon}$, $\Gamma_{\epsilon}$, $\partial \Omega$) will be denoted in the
rest of the paper as {\bf{Assumption 0}}. 
\bigskip

%*** Above, I suppose that $N=3$. 
%Do we really need the statement \ref{a}. 
%Laurent ***
%\bigskip

Throughout this paper, we will abuse notations by denoting by $\epsilon$ a sequence
of positive real numbers which converges to zero.
We will consider in the following a discrete 
coagulation-fragmentation-diffusion model for the evolution of clusters
\cite{4}, \cite{5}.
Denoting by $u_i^{\epsilon}:=u_i^{\epsilon}(t,x) \geq 0$ the density of
clusters with integer size $i \geq 1$ at position $x \in \Omega_{\epsilon}$
and time $t \geq 0$, and by $d_i>0$ the diffusion constant for clusters of size
$i$, the corresponding system can be written as a family of
equations in $\Omega_{\epsilon}$, the first one being:

\begin{eqnarray} \label{1.1}
\begin{cases}
\frac{\displaystyle \partial{u_1^{\epsilon}}}{\displaystyle \partial t}- d_1 \, \Delta_x u_1^{\epsilon}+u_1^{\epsilon} \, \sum_{j=1}^{\infty} 
a_{1,j}
u_j^{\epsilon}=\sum_{j=1}^{\infty} B_{1+j} \, \beta_{1+j,1} \, 
u_{1+j}^{\epsilon}
& \text{in } [0,T] \times \Omega_{\epsilon}, \\

\\ 
\frac{\displaystyle \partial{u_1^{\epsilon}}}{\displaystyle \partial \nu} 
:= \nabla_x u_1^{\epsilon}
\cdot n=0 & \text{on } [0,T] \times \partial\Omega , \\

\\
\frac{\displaystyle \partial{u_1^{\epsilon}}}{\displaystyle \partial \nu}
:= \nabla_x u_1^{\epsilon} \cdot n=
\epsilon \, \psi(t, x, \frac{x}{\epsilon}) & \text{on } 
[0,T] \times \Gamma_{\epsilon} , \\

\\
u_1^{\epsilon}(0,x)=U_1 & \text{in } \Omega_{\epsilon} .
\end{cases}
\end{eqnarray}
We shall systematically make the following assumption on $\psi$ and $U_1$:
\medskip

{\bf{Assumption A}}: We suppose that
$\psi$ is a given (bounded) function satisfying the following conditions:
\par\noindent
(i) $\psi \in C^1 ([0, T]; B)$
with $B=C^1[{\overline \Omega}; C^1_{\#}(Y)]$ ($C^1_{\#}(Y)$ being the space of periodic $C^1$ functions with period relative to $Y$), 
\par\noindent
(ii) $\psi(t=0, x, \frac{x}{\epsilon})=0$ for $x\in \Omega_{\epsilon}$,
\par\noindent
and
$U_1$ is a  constant such that
$0 \le  U_1 \leq \Vert \psi \Vert_{L^{\infty} ([0, T];B)}$.
\medskip

In addition, if $i \geq 2$, we introduce the following equations:

\begin{eqnarray} \label{1.2}
\begin{cases}
\frac{\displaystyle \partial{u_i^{\epsilon}}}{\displaystyle \partial t}- d_i
\,\Delta_x u_i^{\epsilon}
=Q_i^{\epsilon}+F_i^{\epsilon} & \text{in } [0,T] \times \Omega_{\epsilon} ,\\

\\ 
\frac{\displaystyle \partial{u_i^{\epsilon}}}{\displaystyle \partial \nu} 
:= \nabla_x u_i^{\epsilon}
\cdot n=0 & \text{on } [0,T] \times \partial\Omega , \\

\\
\frac{\displaystyle \partial{u_i^{\epsilon}}}{\displaystyle \partial \nu}
:= \nabla_x u_i^{\epsilon} \cdot n=0
& \text{on } 
[0,T] \times \Gamma_{\epsilon} , \\

\\
u_i^{\epsilon}(0,x)=0  & \text{in } \Omega_{\epsilon} , 
\end{cases}
\end{eqnarray}
where the terms $Q_i^{\epsilon}$, $F_i^{\epsilon}$ due to coagulation and
fragmentation, respectively, are given by

\begin{equation} \label{1.3}
Q_i^{\epsilon}:=\frac{1}{2} \, \sum_{j=1}^{i-1} a_{i-j,j} \, 
u_{i-j}^{\epsilon} \, u_j^{\epsilon}-\sum_{j=1}^{\infty} a_{i,j} \,
u_i^{\epsilon} \, u_j^{\epsilon},
\end{equation}

\begin{equation} \label{1.4}
F_i^{\epsilon}:=\sum_{j=1}^{\infty} B_{i+j} \, 
\beta_{i+j,i} \, u_{i+j}^{\epsilon}-B_i \,
u_i^{\epsilon}.
\end{equation}
The parameters $B_i$, $\beta_{i,j}$ and $a_{i,j}$, for integers $i, j \geq 1$,
represent the total rate $B_i$ of fragmentation of clusters of size $i$,
the average number $\beta_{i,j}$ of clusters of size $j$ produced by
fragmentation of a cluster of size $i$, and the coagulation rate $a_{i,j}$ of
clusters of size $i$ with clusters of size $j$.
These parameters represent rates, so they are always nonnegative; single
particles do not break up further, and mass should be conserved when a cluster
breaks up into smaller pieces, so one always imposes the:
\medskip

{\bf{Assumption B}}: The coagulation and fragmentation coefficients satisfy:
\begin{equation} \label{1.5}
a_{i,j}=a_{j,i} \geq 0, \, \; \; \; \beta_{i,j} \geq 0, \, \; \; \; 
(i, j \geq 1),
\end{equation}
\begin{equation} \label{1.6}
B_1=0, \, \; \; \; B_i \geq 0, \, \; \; \; (i \geq 2),
\end{equation}
\begin{equation} \label{1.7}
i=\sum_{j=1}^{i-1} j \, \beta_{i,j}, \, \; \; \;  (i \geq 2).
\end{equation}
In order to prove the bounds presented in the sequel, we need to impose
 additional restrictions on the coagulation  and fragmentation coefficients, together with constraints on the diffusion coefficients. They are summarized in the:
 \medskip
 
 {\bf{Assumption C}}: There exists $C>0$, $\zeta \in ]0,1]$ such that
\begin{equation}  \label{1.40}
a_{ij} \le C\, (i+j)^{1-\zeta} . 
\end{equation}
 Moreover,
for each $m \geq 1$, there exists $\gamma_m >0$ such that
\begin{equation} \label{1.7a}
B_j \, \beta_{j,m} \leq \gamma_m \, a_{m,j} \; \; \; \; \; \; \; 
\text{for} \; \; j \geq m+1.
\end{equation}
%*** We should give a reference in which this condition is used. Laurent ***
%\par 
Finally, there exist constants $D_0, D_1>0$ such
that
\begin{equation}  \label{diff}
\forall i \in \N - \{0\}, \qquad  0 < D_0 \le d_i \le D_1. 
\end{equation}

 Note that the assumption (\ref{1.40}) on the coagulation coefficients $a_{ij}$ is quite standard: it enables to show that no gelation occurs in the considered system of coagulation-fragmentation equations, provided that the diffusion coefficients satisfy the bound (\ref{diff}), cf. \cite{5}. For a set of alternative assumptions (more stringent on the coagulation coefficients, but less stringent on the diffusion coefficients), see \cite{4}. Finally, assumption (\ref{1.7a}) is used in existence proofs for systems where both coagulation and fragmentation are considered, see \cite{22}.

\subsection{Main statement and comments} \label{subsec02}

Our aim is to study the homogenization of the set of equations
(\ref{1.1})-(\ref{1.2}) as ${\epsilon \rightarrow 0}$, i.e., to study
the behaviour of $u_i^{\epsilon}$ $(i \geq 1)$, as
${\epsilon \rightarrow 0}$, and obtain the equations satisfied
by the limit.
Since there is no obvious notion of convergence for the sequence
$u_i^{\epsilon}$ $(i \geq 1)$
(which is defined on a varying set $\Omega_{\epsilon}$:
this difficulty is specific to the case of perforated domains), we use
the natural tool of two-scale convergence as elaborated  by Nguetseng-Allaire, \cite{16}, \cite{3}. Our main statement shows that it is indeed possible to homogenize the equations: 

\begin{theorem} \label{t2.1} For $\epsilon>0$ small enough, there exists  
a strong solution
$$u_i^{\epsilon}:= u_i^\epsilon(t,x)
\in L^2([0,T]; H^2(\Omega_\epsilon)) \cap H^1(]0,T[; L^2(\Omega_\epsilon)) \quad (i \geq 1)
$$
to system \eqref{1.1} - \eqref{1.2}, which is moreover nonnegative, that is 
$$
u_i^\epsilon(t,x)\ge 0 \quad\mbox{for $(t,x)\in (0,T)\times \Omega_\epsilon$}.
$$
We now introduce the notation $\,\,\widetilde{}\,\,$ for the extension by zero outside
$\Omega_{\epsilon}$, and we denote by $\chi := \chi(y)$ the characteristic function of $Y^{\ast}$.
\par 
%We consider $u_i^{\epsilon} := u_i^{\epsilon} (t,x)$ ($i \geq 1$)  a family of
% strong (in the sense defined above) solutions to problems (\ref{1.1})-(\ref{1.2}).
Then the sequences $\widetilde{u_i^{\epsilon}}$ and
$\widetilde{\nabla_x u_i^{\epsilon}}$ ($i \geq 1$)
two-scale converge (up to a subsequence) to
$(t,x,y) \mapsto [\chi(y) \, u_i(t,x)]$ and
$(t,x,y) \mapsto [\chi(y) (\nabla_x u_i(t,x)+\nabla_y u_i^1(t,x,y))]$
($i \geq 1$),
respectively,
% where $\,\,\widetilde{}\,\,$ denotes the extension by zero outside
%$\Omega_{\epsilon}$ and
%$\chi := \chi(y)$ represents the characteristic function of $Y^{\ast}$.
where the limiting functions
$[ (t,x) \mapsto u_i(t,x), (t,x,y) \mapsto u_i^1(t,x,y) ]$ ($i \geq 1$)
are weak solutions lying in
$L^2 (0,T; H^1 (\Omega)) \times 
L^2 ([0,T] \times \Omega; H_{\#}^1(Y)/\mathbb{R})$ of the following two-scale
homogenized systems:

If $i=1$:
\begin{eqnarray} \label{4.1} 
\begin{cases}
\theta \, \frac{\displaystyle \partial u_1}{\displaystyle \partial t}(t,x)
- d_1\, \nabla_x \cdot \bigg[ \, A \, 
\nabla_x u_1(t,x) \bigg]
+\theta \, u_1(t,x)
\sum_{j=1}^{\infty} a_{1,j} \, u_j(t,x)  \\
=\theta \, \sum_{j=1}^{\infty} B_{1+j} \, \beta_{1+j,1} \, u_{1+j}(t,x)+
d_1 \, \displaystyle \int_{\Gamma} \psi(t,x,y) \, d\sigma(y)  
& \text{ in }  [0,T] \times \Omega, \\

\\
[ A \, \nabla_x u_1 (t,x)] \cdot n=0 & \text{ on }  
[0,T] \times \partial \Omega ,\\

\\
u_1(0,x)=U_1 & \text{ in }  \Omega ;
\end{cases}
\end{eqnarray}

If $i \geq 2$:
\begin{eqnarray} \label{4.2} 
\begin{cases}
\theta \, \frac{\displaystyle \partial u_i}{\displaystyle \partial t}(t,x)
- d_i \, \nabla_x \cdot \bigg[ \, A \, 
\nabla_x u_i(t,x) \bigg]
+\theta \, u_i(t,x)
\sum_{j=1}^{\infty} a_{i,j} \, u_j(t,x) \\
+\theta \, B_i \, u_i(t,x)  
=\frac{\displaystyle \theta}{\displaystyle 2} \sum_{j=1}^{i-1} a_{j,{i-j}}
u_j(t,x) \, u_{i-j}(t,x) \\
+\theta \, \sum_{j=1}^{\infty} B_{i+j} \, \beta_{i+j,i} \, u_{i+j}(t,x)
& \text{ in }  [0,T] \times \Omega, \\

\\
[ A \, \nabla_x u_i (t,x)] \cdot n=0 & \text{ on }  
[0,T] \times \partial \Omega ,\\

\\
u_i(0,x)=0 & \text{ in }  \Omega ,
\end{cases}
\end{eqnarray}
where $\theta=\int_{Y} \chi(y) dy=\vert Y^{\ast} \vert$ is
 the volume fraction of material,
and $A$ is a matrix (with constant coefficients) defined by
$$A_{jk}=\displaystyle \int_{Y^{\ast}} (\nabla_y w_j+ \hat{e}_j) \cdot
(\nabla_y w_k+ \hat{e}_k) \, dy , $$
with $\hat{e}_j$ being the $j$-th unit vector in $\mathbb{R}^3$, and
$(w_j)_{1 \leq j \leq 3}$ the family of solutions of the cell
problem
\begin{eqnarray} \label{4.3}
\begin{cases}
-\nabla_y \cdot [\nabla_y w_j+ \hat{e}_j]=0 \, \, \, \qquad \; \; \; \; \; 
\text{ in} \, \, Y^{\ast} ,\\
(\nabla_y w_j+\hat{e}_j) \cdot n=0 \, \, \, \qquad \; \; \; \; \; \; \; \;
\text{ on} \, \, \Gamma ,\\
y \mapsto w_j(y) \; \; \; \; \; Y-\text{periodic} .
\end{cases}
\end{eqnarray}
Finally,
$$u_i^1(t,x,y)=\sum_{j=1}^3 \, w_j (y) \,
\frac{\displaystyle \partial u_i}{\displaystyle \partial x_j}(t,x)
\; \; \; (i \geq 1).$$
\end{theorem}

\subsection{Structure of the rest of the paper}\label{subsec03}

The paper is organized as follows.
In Section \ref{section2}, we derive all the {\it{a priori}} estimates needed for
two-scale homogenization.
In particular, in order to prove the uniform $L^2$-bound on the infinite
sums appearing in our set of Eqs. (\ref{1.1})-(\ref{1.2}), we extend to the
case of non-homogeneous Neumann boundary conditions a duality method first
devised by M. Pierre and D. Schmitt \cite{17} and largely exploited
afterwards \cite{4}, \cite{5}.
Then, Section \ref{section3} is devoted to the proof of our main results
on the homogenization of the infinite Smoluchowski discrete 
coagulation-fragmentation-diffusion equations in a periodically perforated
domain.
Finally, Appendix \ref{appA} and Appendix \ref{appB} are introduced to
summarize, respectively, some fundamental inequalities in Sobolev spaces
tailored for perforated media, and some basic results on the two-scale 
convergence method (used to perform the homogenization procedure).

\section{Estimates}\label{section2}

We first obtain the {\it{a priori}} estimates for the sequences 
$u_i^{\epsilon}$,
$\nabla_x u_i^{\epsilon}$, $\partial_t u_i^{\epsilon}$ in
$[0,T] \times \Omega_{\epsilon}$, that are independent of $\epsilon$. We start with an adapted duality lemma in the style of \cite{17}.

\begin{lemma} \label{l1.1}

Let $\Omega_{\epsilon}$ be an open set satisfying Assumption $0$.
 We suppose that Assumptions A, B, C hold.
%Assume that 
%$$ \sup_i d_i < +\infty.$$
Then, for all $T >0$, classical solutions to system (\ref{1.1})-(\ref{1.2})
satisfy the following bound:

\begin{equation} \label{1.8}
\displaystyle \int_{0}^{T} \displaystyle \int_{\Omega_{\epsilon}}
%\bigg[ \sum_{i=1}^{\infty} i \, d_i \, u_i^{\epsilon}(t, x) \bigg] %\,
\, \bigg[ \sum_{i=1}^{\infty} i \, u_i^{\epsilon}(t, x) \bigg]^2 \, dt \, dx
\leq C ,
\end{equation}
where $C$ is a positive constant independent of $\epsilon$.

\end{lemma}

\begin{proof}

Let us consider the following fundamental identity (or weak formulation) of the
coagulation and fragmentation operators \cite{4}, \cite{5}:

\begin{equation} \label{1.9}
\sum_{i=1}^{\infty} \varphi_i \, Q_i^{\epsilon}= \frac{1}{2} 
\sum_{i=1}^{\infty} \, \sum_{j=1}^{\infty} a_{i,j} \, u_i^{\epsilon} \, 
u_j^{\epsilon} \, (\varphi_{i+j}-\varphi_i-\varphi_j) ,
\end{equation}

\begin{equation} \label{1.10}
\sum_{i=1}^{\infty} \varphi_i \, F_i^{\epsilon}=-\sum_{i=2}^{\infty}  B_i \,
u_i^{\epsilon} \bigg( \varphi_i-\sum_{j=1}^{i-1} \beta_{i,j} \, \varphi_j 
\bigg),
\end{equation}
which holds for any sequence of numbers $(\varphi_i)_{i \ge 1}$ such that the sums are defined.
By choosing $\varphi_i :=i$ above and thanks to (\ref{1.7}), we have the mass conservation property for the operators $Q_i^{\epsilon}$ and $F_i^{\epsilon}$:
\begin{equation} \label{1.11}
\sum_{i=1}^{\infty} i \, Q_i^{\epsilon}=\sum_{i=1}^{\infty} i \, F_i^{\epsilon}
=0.
\end{equation}
Therefore, summing together Eq. (\ref{1.1}) and Eq. (\ref{1.2}) multiplied by
$i$, taking into account the identity (\ref{1.11}), we get the (local in $x$) mass conservation property for the system:
\begin{equation} \label{1.12}
\frac{\partial}{\partial t} \bigg[ \sum_{i=1}^{\infty} i \, u_i^{\epsilon}
\bigg]-\Delta_x \bigg[ \sum_{i=1}^{\infty} i \, d_i \, u_i^{\epsilon} \bigg]=0.
\end{equation}
Denoting
\begin{equation} \label{1.13}
\rho^{\epsilon}(t,x)=\sum_{i=1}^{\infty} i \, u_i^{\epsilon}(t,x),
\end{equation}
and
\begin{equation} \label{1.14}
A^{\epsilon}(t,x)=[\rho^{\epsilon}(t,x)]^{-1} \sum_{i=1}^{\infty} i \, d_i \,
u_i^{\epsilon}(t,x) ,
\end{equation}
the following system can be derived from Eqs. (\ref{1.1}), (\ref{1.2}) and
(\ref{1.12}):

\begin{eqnarray} \label{1.15}
\begin{cases}
\frac{\displaystyle \partial{\rho^{\epsilon}}}{\displaystyle \partial t}-
\Delta_x (A^{\epsilon} \, \rho^{\epsilon})=0
& \text{in } [0,T] \times \Omega_{\epsilon}, \\
\\ 
\nabla_x (A^{\epsilon} \, \rho^{\epsilon})
\cdot n=0 & \text{on } [0,T] \times \partial\Omega, \\
\\
\nabla_x (A^{\epsilon} \, \rho^{\epsilon}) \cdot n=
d_1 \, \epsilon \, \psi(t, x, \frac{x}{\epsilon}) & \text{on } 
[0,T] \times \Gamma_{\epsilon} , \\
\\
\rho^{\epsilon}(0,x)=U_1 & \text{in } \Omega_{\epsilon} .
\end{cases}
\end{eqnarray}
We observe that (for all $t \in [0, T]$)
\begin{equation} \label{1.16}
\Vert A^{\epsilon}(t, \cdot) \Vert_{L^{\infty}(\Omega_{\epsilon})} \leq
\sup_i \, d_i.
\end{equation}
Multiplying the first equation in (\ref{1.15}) by the function $w^{\epsilon}$
defined by the following dual problem:

\begin{eqnarray} \label{1.17}
\begin{cases}
-\bigg( \frac{\displaystyle \partial{w^{\epsilon}}}{\displaystyle \partial t}+
A^{\epsilon} \, \Delta_x w^{\epsilon} \bigg)=A^{\epsilon} \, \rho^{\epsilon}
& \text{in } [0,T] \times \Omega_{\epsilon}, \\
\\ 
\nabla_x w^{\epsilon}
\cdot n=0 & \text{on } [0,T] \times \partial\Omega, \\
\\
\nabla_x w^{\epsilon} \cdot n=0 & \text{on } 
[0,T] \times \Gamma_{\epsilon} , \\
\\
w^{\epsilon}(T,x)=0 & \text{in } \Omega_{\epsilon} ,
\end{cases}
\end{eqnarray}
and integrating by parts on $[0,T] \times \Omega_{\epsilon}$, we end up with
the identity
\begin{equation} \label{1.18} 
\begin{split}
&\displaystyle \int_{0}^{T} \, \displaystyle \int_{\Omega_{\epsilon}}
A^{\epsilon}(t,x) \, (\rho^{\epsilon}(t,x))^2 \, dt \, dx=
\displaystyle \int_{\Omega_{\epsilon}} w^{\epsilon}(0,x) \, 
\rho^{\epsilon}(0,x) \, dx \\
&+\epsilon \, d_1 \, \displaystyle \int_{0}^{T} \, 
\displaystyle \int_{\Gamma_{\epsilon}} \psi(t, x, \frac{x}{\epsilon}) \,
w^{\epsilon}(t,x) \, dt \, d\sigma_{\epsilon}(x):=I_1+I_2,
\end{split}
\end{equation}
where $d\sigma_{\epsilon}$ is the measure on $\Gamma_{\epsilon}$.

Let us now estimate the terms $I_1$ and $I_2$.
From H\"older's inequality we obtain
\begin{equation} \label{1.19}
I_1=\displaystyle \int_{\Omega_{\epsilon}} w^{\epsilon}(0,x)  \, 
\rho^{\epsilon}(0,x) \, dx \leq U_1 \, \vert \Omega_{\epsilon} \vert^{1/2} \,
\Vert w^{\epsilon}(0, \cdot) \Vert_{L^2(\Omega_{\epsilon})}.
\end{equation}
Applying once more H\"older's inequality and using  estimate (\ref{1.16}), we get
\begin{equation} \label{1.20}
\begin{split}
&\displaystyle \int_{\Omega_{\epsilon}} \vert w^{\epsilon}(0,x) \vert^2 \,
dx=\displaystyle \int_{\Omega_{\epsilon}} 
\bigg \vert \displaystyle \int_{0}^{T} \sqrt{A^{\epsilon}} \,
\frac{\displaystyle \partial_t \, w^{\epsilon}}
{\displaystyle \sqrt{A^{\epsilon}}} \, dt \bigg \vert^2 \, dx \\ 
&\leq T \, \Vert A^{\epsilon} \Vert_{L^{\infty}(\Omega_{\epsilon})}
\displaystyle \int_{0}^{T} \, \displaystyle \int_{\Omega_{\epsilon}}
(A^{\epsilon})^{-1} \, \bigg \vert \frac{\partial}{\partial t} 
w^{\epsilon}(t,x) \bigg \vert^2 \, dt \, dx \\
& \leq T \, (\sup_i \, d_i) \displaystyle \int_{0}^{T} 
\displaystyle \int_{\Omega_{\epsilon}} (A^{\epsilon})^{-1} \,
\bigg \vert \frac{\partial}{\partial t} w^{\epsilon}(t,x) \bigg \vert^2 \, 
dt \, dx.
\end{split}
\end{equation}
By exploiting the dual problem (\ref{1.17}), Eq. (\ref{1.20}) becomes
\begin{equation} \label{1.21}
\begin{split}
&\displaystyle \int_{\Omega_{\epsilon}} \vert w^{\epsilon}(0,x) \vert^2 \,
dx \leq T \, (\sup_i \, d_i) \displaystyle \int_{0}^{T}
\displaystyle \int_{\Omega_{\epsilon}} (A^{\epsilon})^{-1} \,
\vert A^{\epsilon} \, \Delta_x w^{\epsilon}+A^{\epsilon} \, \rho^{\epsilon} 
\vert^2 \, dt \, dx \\
&\leq T \, (\sup_i \, d_i) \displaystyle \int_{0}^{T}
\displaystyle \int_{\Omega_{\epsilon}} (A^{\epsilon})^{-1} \,
\bigg[ 2 \, (A^{\epsilon})^2 \, (\Delta_x w^{\epsilon})^2+
2 \, (A^{\epsilon})^2 \, (\rho^{\epsilon})^2 \bigg] \, dt \, dx.
\end{split}
\end{equation}
Let us now estimate the first term on the right-hand side of (\ref{1.21}).
Multiplying the first equation in (\ref{1.17}) by ($\Delta_x w^{\epsilon}$), we see that
\begin{equation} \label{1.22}
\displaystyle \int_{\Omega_{\epsilon}} (\Delta_x w^{\epsilon}) \,
\bigg( \frac{\displaystyle \partial w^{\epsilon}}
{\displaystyle \partial t} \bigg) \, dx + 
\displaystyle \int_{\Omega_{\epsilon}} A^{\epsilon} \, 
(\Delta_x w^{\epsilon})^2 \, dx=-\displaystyle \int_{\Omega_{\epsilon}}
A^{\epsilon} \, \rho^{\epsilon} \, (\Delta_x w^{\epsilon}) \, dx ,
\end{equation}
and integrating by parts on $\Omega_{\epsilon}$, we get
\begin{equation} \label{1.23}
-\frac{\displaystyle \partial}{\displaystyle \partial t}
\displaystyle \int_{\Omega_{\epsilon}} \frac{\displaystyle \vert \nabla_x 
 w^{\epsilon} \vert^2}{\displaystyle 2} \, dx+ 
\displaystyle \int_{\Omega_{\epsilon}} A^{\epsilon} \, 
(\Delta_x w^{\epsilon})^2 \, dx=-\displaystyle \int_{\Omega_{\epsilon}}
A^{\epsilon} \, \rho^{\epsilon} (\Delta_x w^{\epsilon}) \, dx.
\end{equation}
Then, integrating once more over the time interval $[0,T]$ and using 
Young's inequality for the right-hand side of (\ref{1.23}), one finds that
\begin{equation} \label{1.24}
\displaystyle \int_{\Omega_{\epsilon}} \vert \nabla_x w^{\epsilon}(0,x) \vert^2
\, dx+\displaystyle \int_{0}^{T} \displaystyle \int_{\Omega_{\epsilon}}
A^{\epsilon} \, (\Delta_x w^{\epsilon})^2 \, dt \, dx \leq
\displaystyle \int_{0}^{T} \displaystyle \int_{\Omega_{\epsilon}}
(\rho^{\epsilon})^2 \, A^{\epsilon} \, dt \, dx.
\end{equation}
Since the first term of the left-hand side of (\ref{1.24}) is nonnegative,
we conclude that
\begin{equation} \label{1.25}
\displaystyle \int_{0}^{T} \displaystyle \int_{\Omega_{\epsilon}}
A^{\epsilon} \, (\Delta_x w^{\epsilon})^2 \, dt \, dx \leq
\displaystyle \int_{0}^{T} \displaystyle \int_{\Omega_{\epsilon}}
(\rho^{\epsilon})^2 \, A^{\epsilon} \, dt \, dx.
\end{equation}
Inserting Eq. (\ref{1.25}) into Eq. (\ref{1.21}), one obtains
\begin{equation} \label{1.26}
\displaystyle \int_{\Omega_{\epsilon}} \vert w^{\epsilon}(0,x) \vert^2 \, dx
\leq 4 \, T \, (\sup_i \, d_i) \displaystyle \int_{0}^{T}
\displaystyle \int_{\Omega_{\epsilon}} A^{\epsilon} \, (\rho^{\epsilon})^2 \, 
dt \, dx.
\end{equation}
Therefore, we end up with the estimate
\begin{equation} \label{1.27}
I_1 \leq 2 \, U_1 \, \bigg[ \vert \Omega_{\epsilon} \vert \, T \, \sup_i \, d_i
\bigg]^{1/2} \, \bigg[ \displaystyle \int_{0}^{T} 
\displaystyle \int_{\Omega_{\epsilon}} A^{\epsilon} \, (\rho^{\epsilon})^2
\, dt \, dx \bigg]^{1/2}.
\end{equation}
By using Lemma \ref{lA.0} of Appendix A and H\"older's inequality, the term $I_2$ in
(\ref{1.18}) can be rewritten as
\begin{equation} \label{1.28}
\begin{split}
&I_2=\epsilon \, d_1 \displaystyle \int_{0}^{T}
\displaystyle \int_{\Gamma_{\epsilon}} \psi(t, x, \frac{x}{\epsilon}) \,
w^{\epsilon}(t,x) \, dt \, d\sigma_{\epsilon}(x) \\
&\leq \sqrt{C_1\, \tilde{C}} \, d_1 \displaystyle \int_{0}^{T} \Vert \psi(t) \Vert_B \bigg \{
\bigg[ \displaystyle \int_{\Omega_{\epsilon}} \vert w^{\epsilon} \vert^2
\, dx \bigg]^{1/2}+\epsilon \bigg[ \displaystyle \int_{\Omega_{\epsilon}}
\vert \nabla_x w^{\epsilon} \vert^2 \, dx \bigg]^{1/2} \bigg \}, 
\end{split}
\end{equation}
where we have taken into account the following estimate (see Lemma \ref{lB.6} of Appendix B):

\begin{equation} \label{1.29}
\epsilon \, \displaystyle \int_{\Gamma_{\epsilon}}
\vert \psi (t,x,\frac{x}{\epsilon} ) \vert^2 \, d\sigma_{\epsilon}(x)
\leq {\tilde C} \, \Vert \psi (t) \Vert^2_B 
\end{equation}
(with ${\tilde C}$ being a positive constant independent of $\epsilon$ and
$B=C^1 [\overline{\Omega}; C^1_{\#} (Y)]$).
Note that we do not really need that $\psi$ be of class $C^1$ in the estimate above, 
continuity would indeed be sufficient.
\medskip

Since $\psi \in L^{\infty}([0,T];B)$, using the Cauchy-Schwarz inequality,
Eq. (\ref{1.28}) reads
\begin{equation} \label{1.30}
I_2 \leq C \, d_1 \, \Vert w^{\epsilon} 
\Vert_{L^2(0,T; L^2(\Omega_{\epsilon}))}
+C \, d_1 \, \epsilon \,
\Vert \nabla_x w^{\epsilon} \Vert_{L^2(0,T; L^2(\Omega_{\epsilon}))}:=J_1+J_2,
\end{equation}
where $C>0$ is a constant independent of $\epsilon$.
Let us now estimate the terms $J_1$ and $J_2$.
Using H\"older's inequality and estimate (\ref{1.16}), by following the
same strategy as the one leading  to (\ref{1.26}), we get
\begin{equation} \label{1.31}
\begin{split}
&\displaystyle \int_{0}^{T} \displaystyle \int_{\Omega_{\epsilon}}
\vert w^{\epsilon}(t,x) \vert^2 \, dt \, dx =
\displaystyle \int_{0}^{T} \displaystyle \int_{\Omega_{\epsilon}}
\bigg \vert \displaystyle \int_{t}^{T} \sqrt{A^{\epsilon}} \,
\frac{\displaystyle \partial_s w^{\epsilon}(s,x)}
{\displaystyle \sqrt{A^{\epsilon}}} \, ds \bigg \vert^2 \, dt \, dx \\
&\leq T^2 \, (\sup_i \, d_i) \displaystyle \int_{0}^{T} 
\displaystyle \int_{\Omega_{\epsilon}} (A^{\epsilon})^{-1} \,
\bigg \vert \frac{\displaystyle \partial w^{\epsilon}}
{\displaystyle \partial t} (t,x) \bigg \vert^2 \, dt \, dx \\
&\leq 
4 \, T^2 \, (\sup_i \, d_i) \displaystyle \int_{0}^{T}
\displaystyle \int_{\Omega_{\epsilon}} A^{\epsilon} \, (\rho^{\epsilon})^2
\, dt \, dx,
\end{split}
\end{equation}
so that
\begin{equation} \label{1.32}
\begin{split}
&J_1=C \, d_1 \bigg[ \displaystyle \int_{0}^{T} \displaystyle 
\int_{\Omega_{\epsilon}} \vert w^{\epsilon}(t,x) \vert^2 \, dt \, dx 
\bigg]^{1/2} \\
&\leq 2 \, C \, d_1 \, T \, (\sup_i \, d_i)^{1/2} \bigg[
\displaystyle \int_{0}^{T} \displaystyle \int_{\Omega_{\epsilon}} 
A^{\epsilon} \, (\rho^{\epsilon})^2 \, dt \, dx \bigg]^{1/2}.
\end{split}
\end{equation}
In order to estimate $J_2$, we go back to Eq. (\ref{1.23}).
Integrating over $[t,T]$, one obtains
\begin{equation} \label{1.33}
\begin{split}
&\frac{\displaystyle 1}{\displaystyle 2} \displaystyle \int_{t}^{T}
\displaystyle \int_{\Omega_{\epsilon}} \frac{\partial}{\partial s}
\vert \nabla_x w^{\epsilon}(s,x) \vert^2 \, ds \, dx -
 \displaystyle \int_{t}^{T} \displaystyle \int_{\Omega_{\epsilon}}
A^{\epsilon} \, (\Delta_x w^{\epsilon})^2 \, ds \, dx \\
&=\displaystyle \int_{t}^{T} \displaystyle \int_{\Omega_{\epsilon}}
A^{\epsilon} \, \rho^{\epsilon} \, (\Delta_x w^{\epsilon}) \, ds \, dx.
\end{split}
\end{equation}
 Young's inequality applied to the right-hand side of Eq. (\ref{1.33}) leads to
\begin{equation} \label{1.34}
\displaystyle \int_{\Omega_{\epsilon}} \vert \nabla_x w^{\epsilon}(t,x) 
\vert^2 \, dx+ \displaystyle \int_{t}^{T}
\displaystyle \int_{\Omega_{\epsilon}} A^{\epsilon} \, 
(\Delta_x w^{\epsilon})^2 \, ds \, dx \leq \displaystyle \int_{t}^{T}
\displaystyle \int_{\Omega_{\epsilon}} A^{\epsilon} \, (\rho^{\epsilon})^2
\, ds \, dx.
\end{equation}
Taking into account that the second term on the left-hand side of 
(\ref{1.34}) is nonnegative and integrating once more over time, we get
\begin{equation} \label{1.35}
\displaystyle \int_{0}^{T} \displaystyle \int_{\Omega_{\epsilon}}
\vert \nabla_x w^{\epsilon}(t,x) \vert^2 \, dt \, dx \leq T \,
\displaystyle \int_{0}^{T} \displaystyle \int_{\Omega_{\epsilon}}
A^{\epsilon} \, (\rho^{\epsilon})^2 \, dt \, dx.
\end{equation}
Therefore, we conclude that
\begin{equation} \label{1.36}
\begin{split}
J_2=&C \, d_1 \, \epsilon \, \bigg[ \displaystyle \int_{0}^{T}
\displaystyle \int_{\Omega_{\epsilon}} \vert \nabla_x w^{\epsilon}(t,x) 
\vert^2 \, dt \, dx \bigg]^{1/2} \\
&\leq C \, d_1 \, \epsilon \, (T)^{1/2} \, \bigg[ 
\displaystyle \int_{0}^{T} \displaystyle \int_{\Omega_{\epsilon}}
A^{\epsilon} \, (\rho^{\epsilon})^2 \, dt \, dx \bigg]^{1/2}.
\end{split}
\end{equation}
By combining (\ref{1.32}) and (\ref{1.36}), we end up with the estimate

\begin{equation} \label{1.37}
I_2 \leq d_1 \, \bigg[ 2 \, C \, T \, \sqrt{\sup_i \, d_i}+
C \, \epsilon \,
\sqrt{T} \bigg] \, \bigg[ \displaystyle \int_{0}^{T}
\displaystyle \int_{\Omega_{\epsilon}} A^{\epsilon} \, (\rho^{\epsilon})^2
\, dt \, dx \bigg]^{1/2}.
\end{equation}
Hence, inserting estimates (\ref{1.27}) and (\ref{1.37}) in 
Eq. (\ref{1.18}), one obtains
\begin{equation} \label{1.38}
 \displaystyle \int_{0}^{T} \displaystyle \int_{\Omega_{\epsilon}}
A^{\epsilon}(t,x) \, (\rho^{\epsilon}(t,x))^2 \, dt \, dx 
\leq C_3^2 ,
\end{equation}
where
\begin{equation} \label{1.39}
C_3=\max \bigg( 2 \, U_1 \, \sqrt{\vert \Omega_{\epsilon} \vert \, T \,
\sup_i \, d_i}, d_1 \, [2 \, C \, T \, \sqrt{\sup_i \, d_i}+C \,
\sqrt{T}] \bigg).
\end{equation}
Thus, recalling the definitions of $A^{\epsilon}$ and $\rho^{\epsilon}$, and using the lower bound on the diffusion rates in Assumption C,
the assertion of the Lemma  immediately follows.

\end{proof}

\begin{corollary} \label{cor22}
 Let $\Omega_{\epsilon}$ be an open set satisfying Assumption $0$.
 Under Assumptions A, B and C, 
%and
%\begin{equation} \label{1.40}
%a_{i,j} \leq \text{Cst} \, (i^{1 - \zeta} + j^{1 - \zeta}) \; 
%\end{equation}
%for some $\zeta>0$, then, the estimate (\ref{1.8}) leads to 
the following bound holds for all classical solutions of (\ref{1.1}), (\ref{1.2}), when $i \ge 1$:
\begin{equation} \label{1.41}
\displaystyle \int_{0}^{T} \displaystyle \int_{\Omega_{\epsilon}}
\bigg \vert \sum_{j=1}^{\infty}  a_{i,j} \, u_j^{\epsilon}(t,x) \bigg \vert^2
\, dt \, dx \leq C_i,
\end{equation}
where $C_i$ does not depend on $\epsilon$ (but may depend on $i$).
\end{corollary}

\begin{proof} Thanks to estimate (\ref{1.8}), we see that
$$ \displaystyle \int_{0}^{T} \displaystyle \int_{\Omega_{\epsilon}}
\bigg \vert \sum_{j=1}^{\infty} j\, u_j^{\epsilon}(t,x) \bigg \vert^2
\, dt \, dx \leq C. $$
We conclude using estimate (\ref{1.40}) of Assumption C.
\end{proof} 

\begin{remark}
We first notice that in order to get Corollary \ref{cor22} (and the results of this section which use it),
it would be sufficient to assume that $a_{i,j} \le C\, (i+j)$. 
We will however need the more stringent estimate (\ref{1.40}) of Assumption C in the proof of the homogenization result in next section. Note that
this Assumption ensures that no gelation occurs in the coagulation-fragmentation process that we consider (cf. \cite{5}).
\par 
We also could relax the hypothesis that the diffusion rates $d_i$ be bounded below (and replace it by the assumption that $d_i$ behaves as a (negative) power law), provided that the assumption on the growth 
coefficients $a_{i,j}$ be made more stringent (cf. \cite{4}).
In that situation, the duality lemma reads
$$ \displaystyle \int_{0}^{T} \displaystyle \int_{\Omega_{\epsilon}}
\bigg[ \sum_{i=1}^{\infty} i \, d_i \, u_i^{\epsilon}(t, x) \bigg] \,
\, \bigg[ \sum_{i=1}^{\infty} i \, u_i^{\epsilon}(t, x) \bigg] \, dt \, dx
\leq C . $$
\end{remark}

We now turn to $L^{\infty}$ estimates. We start with the

\begin{lemma} \label{l1.2}
 Let $\Omega_{\epsilon}$ be an open set satisfying Assumption $0$.
We also suppose that Assmptions A, B, and C hold. We finally
consider $T >0$,  and a classical 
solution $u_i^{\epsilon}$ ($i\ge 1$), of (\ref{1.1}) - (\ref{1.2}).
Then, the following estimate holds:
\begin{equation} \label{1.42} 
{\Vert u_1^{\epsilon} \Vert}_{L^{\infty}
(0, T; L^{\infty} (\Omega_{\epsilon}))}
\leq 
|U_1|+ \Vert u_1^\epsilon \Vert_{L^{\infty} (0, T ; 
L^{\infty} (\Gamma_{\epsilon}))} +  \gamma_1 + 1.
\end{equation}

\end{lemma}

\begin{proof}

Let us test the first equation of (\ref{1.1}) with the function
$$\phi_1 := p \, (u_1^{\epsilon})^{(p-1)} \; \; \; \; p \geq 2.$$
We stress that the function $\phi_1 $ is strictly positive  and continuously
differentiable on $[0,t] \times \overline\Omega$, for all $t>0$.
Integrating, the divergence theorem yields

\begin{equation} \label{1.43}
\begin{split}
&\displaystyle \int_{0}^{t} ds \, \displaystyle \int_{\Omega_{\epsilon}}
\frac{\partial}{\partial s} (u_1^{\epsilon})^{p} (s) \, dx +
d_1 \, p \, (p-1) \, \displaystyle \int_{0}^{t} ds \,
\displaystyle \int_{\Omega_{\epsilon}} \vert \nabla_x u_1^{\epsilon}
\vert^2 \, (u_1^{\epsilon})^{(p-2)} \, dx \\
&=-p \, \displaystyle \int_{0}^{t} ds \, \displaystyle \int_{\Omega_{\epsilon}}
a_{1,1} \, (u_1^{\epsilon})^{(p+1)} \, dx 
-p \, \displaystyle \int_{0}^{t} ds \, \displaystyle \int_{\Omega_{\epsilon}}
(u_1^{\epsilon})^{p} \, \sum_{j=2}^{\infty} a_{1,j} \, u_j^{\epsilon} \,dx \\
& +p \, \displaystyle \int_{0}^{t} ds \, \displaystyle \int_{\Omega_{\epsilon}}
(u_1^{\epsilon})^{(p-1)} \, \sum_{j=2}^{\infty} B_j \, \beta_{j,1} \,
 u_j^{\epsilon} \,dx +
\epsilon \, d_1 \, p \, \displaystyle \int_{0}^{t} ds \, 
\displaystyle \int_{\Gamma_{\epsilon}} \psi(s, x, \frac{x}{\epsilon}) \,
(u_1^{\epsilon})^{(p-1)} \, d\sigma_{\epsilon}(x) \\
& \leq -p \, \displaystyle \int_{0}^{t} ds \, 
\displaystyle \int_{\Omega_{\epsilon}} \sum_{j=2}^{\infty} [a_{1,j} \, 
u_1^{\epsilon}-B_j \, \beta_{j,1}] \, u_j^{\epsilon} \, 
(u_1^{\epsilon})^{(p-1)} \, dx \\
&+\epsilon \, d_1 \, p \, \displaystyle \int_{0}^{t} ds \,
\displaystyle \int_{\Gamma_{\epsilon}} \psi(s, x, \frac{x}{\epsilon}) \,
(u_1^{\epsilon})^{(p-1)} \, d\sigma_{\epsilon}(x).
\end{split}
\end{equation}
Exploiting Assumption C,
% (\ref{1.7a}),
 we end up with the estimate
\begin{equation} \label{1.44}
\begin{split}
&\displaystyle \int_{0}^{t} ds \, \displaystyle \int_{\Omega_{\epsilon}}
\frac{\partial}{\partial s} (u_1^{\epsilon})^{p} (s) \, dx +
d_1 \, p \, (p-1) \, \displaystyle \int_{0}^{t} ds \,
\displaystyle \int_{\Omega_{\epsilon}} \vert \nabla_x u_1^{\epsilon}
\vert^2 \, (u_1^{\epsilon})^{(p-2)} \, dx \\
&\leq \epsilon \, d_1 \, p \, \displaystyle \int_{0}^{t} ds \,
\displaystyle \int_{\Gamma_{\epsilon}} \psi(s, x, \frac{x}{\epsilon}) \,
(u_1^{\epsilon})^{(p-1)} \, d\sigma_{\epsilon}(x)\\
& + p\,\gamma_1^p \int_{0}^{t} ds \, \displaystyle \int_{\Omega_{\epsilon}} 
\sum_{j=2}^{\infty} a_{1,j}\,  u_j^{\epsilon} \,\, dx .
\end{split}
\end{equation}
H\"older's inequality applied to the right-hand side of (\ref{1.44}),
together with  the duality estimate  (\ref{1.41}), leads
to

\begin{equation} \label{1.45}
\begin{split}
&\displaystyle \int_{\Omega_{\epsilon}} (u_1^{\epsilon}(t,x))^{p} \,dx
+d_1 \, p \, (p-1) \, \displaystyle \int_{0}^{t} ds \,
\displaystyle \int_{\Omega_{\epsilon}} \vert \nabla_x u_1^{\epsilon}
\vert^2 \, (u_1^{\epsilon})^{(p-2)} \, dx \\
&\leq \displaystyle \int_{\Omega_{\epsilon}} U_1^p \, dx+
\epsilon \, d_1 \, p \, \Vert \psi \Vert_{L^{\infty} (0,T; 
L^{\infty}(\Gamma_{\epsilon}))} \, \displaystyle \int_{0}^{t} ds \,
\displaystyle \int_{\Gamma_{\epsilon}} (u_1^{\epsilon})^{(p-1)} \, 
d\sigma_{\epsilon}(x) + C\,  p\,\gamma_1^p\,
|\Omega_{\epsilon}|^{1/2} . 
\end{split}
\end{equation}
Since the second term of the left-hand side of (\ref{1.45}) is nonnegative,
one gets
\begin{equation} \label{1.46}
\begin{split}
&\displaystyle \int_{\Omega_{\epsilon}} (u_1^{\epsilon}(t,x))^{p} \,dx
\leq \displaystyle \int_{\Omega_{\epsilon}} U_1^p \, dx \\
&+\epsilon \, d_1 \, p \, \Vert \psi \Vert_{L^{\infty} (0,T;
L^{\infty}(\Gamma_{\epsilon}))} \, \displaystyle \int_{0}^{t} ds \,
\displaystyle \int_{\Gamma_{\epsilon}} [1+(u_1^{\epsilon})^{p}] \,
d\sigma_{\epsilon}(x) + C\,  p\,\gamma_1^p\, |\Omega_{\epsilon}|^{1/2} \\
&\leq \displaystyle \int_{\Omega_{\epsilon}} U_1^p \, dx +
\epsilon \, d_1 \, p \, \Vert \psi \Vert_{L^{\infty} (0,T;
L^{\infty}(\Gamma_{\epsilon}))} \, T \, \vert \Gamma_{\epsilon} \vert \\
&+\epsilon \, d_1 \, p \, \Vert \psi \Vert_{L^{\infty} (0,T;
L^{\infty}(\Gamma_{\epsilon}))} \, \displaystyle \int_{0}^{t} ds \,
\displaystyle \int_{\Gamma_{\epsilon}} (u_1^{\epsilon})^{p} \,
d\sigma_{\epsilon}(x) + C\,  p\,\gamma_1^p\,
|\Omega|^{1/2}.
\end{split}
\end{equation}
Hence, we conclude that
\begin{equation} \label{1.47}
\sup_{t \in [0, T]} \lim_{p\to\infty} \bigg[
\displaystyle \int_{\Omega_{\epsilon}} (u_1^{\epsilon} (t, x))^{p} \,
dx \bigg]^{1/p} \leq \vert U_1 \vert +
\Vert u_1^\epsilon \Vert_{L^{\infty} (0, T ; 
L^{\infty} (\Gamma_{\epsilon}))} +  \gamma_1 + 1.
\end{equation}
\end{proof}
The boundedness of $u_1^{\epsilon}$ in $L^{\infty} ([0, T] \times
\Gamma_{\epsilon})$, uniformly in $\epsilon$, can then be immediately deduced
from Lemma \ref{l1.3} below.

\begin{lemma} \label{l1.3}
Let $\Omega_{\epsilon}$ be an open set satisfying Assumption $0$.
We also suppose that Assmptions A, B, and C hold. We finally
consider $T >0$,  and a classical 
solution $u_i^{\epsilon}$ ($i\ge 1$), of (\ref{1.1}) - (\ref{1.2}).
%Then, the following estimate holds:
Then, for $\epsilon >0$ small enough,
\begin{equation} \label{1.48} 
{\Vert u_1^{\epsilon} \Vert}_{L^{\infty}
(0, T; L^{\infty} (\Gamma_{\epsilon}))}
\leq C,
\end{equation}
where $C$ does not depend on $\epsilon$.

\end{lemma}
In order to establish Lemma \ref{l1.3}, we will first need the following
preliminary result, proven in \cite{10}:

\begin{proposition}[\cite{10}, Theorem 5.2, p.730-732] \label{t1.1}
Let $\Omega_{\epsilon}$ be an open set satisfying Assumption $0$, and $T>0$. We consider a sequence $w^{\epsilon}: = w^{\epsilon}(t,x)\ge 0$
defined on $[0,T] \times \Omega_{\epsilon}$ such that, for some  
$\hat{k}>0$, $\beta>0$, 
%Assume that there exist positive constants $T$,
%$\hat{k}=\Vert \psi \Vert_{L^{\infty} (0, T; 
%B)}$, $\gamma$,
and all
$k \geq \hat{k}$, 
\begin{equation} \label{2.1} 
\Vert (w^{\epsilon} - k)_{+} \Vert^{2}_{Q_{\epsilon} (T)}:=
\sup_{0 \leq t \leq T}
\displaystyle \int_{\Omega_{\epsilon}} \vert (w^{\epsilon} - k)_{+} \vert^2 \, dx +
\displaystyle \int_{0}^{T} \, dt \displaystyle \int_{\Omega_{\epsilon}}
\vert \nabla [(w^{\epsilon} - k)_{+}] \vert^2 \, dx
\end{equation}
$$\leq
 \epsilon \, \beta \, k^2 \, 
\displaystyle \int_{0}^{T}  dt \, \int_{\Gamma_{\epsilon}} 1_{\{w^{\epsilon} > k\} }\, dx .
%\vert B_{k}^{\epsilon} (t) \vert .
 $$
%where $u_{\epsilon}^{(k)} (t):= (u_1^{\epsilon} (t)-k)_{+}$ and
%$B_{k}^{\epsilon} (t)$ is the set of points on $\Gamma_{\epsilon}$
%at which $u_1^{\epsilon} (t, x) >k$.

Then
\begin{equation} \label{2.2}
%ess \, sup_{(t, x) \in [0, T] \times \Gamma_{\epsilon}} 
 ||w^{\epsilon}||_{L^{\infty}([0, T] \times \Gamma_{\epsilon} )} \leq   C(\beta, T, \Omega) \, \hat{k} ,
\end{equation}
where the positive constant $C(\beta, T, \Omega)$ may depend on 
$\beta$, but not on $\hat{k}$ and $\epsilon$.

\end{proposition}
\medskip

%{\bf{Lemma and Theorem above to be precised. Laurent}}
%\medskip

\begin{proof} {\it of Lemma \ref{l1.3}} :

Since this proof is close to the proof of Lemma $5.2$ in \cite{10}, we only
sketch it.
Let $T >0$ and $k \geq 0$ be fixed.
We define: $u_{\epsilon}^{(k)} (t):= (u_1^{\epsilon} (t)-k)_{+}$ for $t \geq 0$.
Its derivatives are
\begin{equation} \label{2.3}
\frac{\displaystyle \partial u_{\epsilon}^{(k)}}{\displaystyle \partial t}=
\frac{\displaystyle \partial u_1^{\epsilon}}{\displaystyle \partial t} \,
\mathbb{1}_{\{u_1^{\epsilon} >k\}},
\end{equation}
\begin{equation} \label{2.4}
\nabla_x u_{\epsilon}^{(k)}= \nabla_x u_1^{\epsilon} \,
\mathbb{1}_{\{u_1^{\epsilon} >k\}}.
\end{equation}
Moreover,
\begin{equation} \label{2.5}
u_{\epsilon}^{(k)} \mid_{\partial \Omega}=
(u_1^{\epsilon} \mid_{\partial \Omega} -k)_{+},
\end{equation}
\begin{equation} \label{2.6}
u_{\epsilon}^{(k)} \mid_{\Gamma_{\epsilon}}=
(u_1^{\epsilon} \mid_{\Gamma_{\epsilon}} -k)_{+} .
\end{equation}
We define $\hat{k} := \max(\Vert \psi
\Vert_{L^{\infty} (0, T; B)}, \gamma_1)$, and consider
% Let us assume 
 $k \geq \hat{k}$.
 % where $\hat{k} := \Vert \psi
%\Vert_{L^{\infty} (0, T; B)}.$
Then, 
\begin{equation} \label{2.7}
u_1^{\epsilon} (0, x)=U_1 \leq \hat{k} \leq k.
\end{equation}
For $t \in [0, T_1]$ with $T_1\le T$, we get therefore
\begin{equation} \label{2.8} 
\begin{split}
\frac{\displaystyle 1}{\displaystyle 2}
\displaystyle \int_{\Omega_{\epsilon}} \vert u_{\epsilon}^{(k)} (t) \vert^2 \,
dx &= \displaystyle \int_{0}^{t} \frac{\displaystyle d}{\displaystyle ds}
\bigg[  \frac{\displaystyle 1}{\displaystyle 2}
\displaystyle \int_{\Omega_{\epsilon}}  
\vert u_{\epsilon}^{(k)} (s) \vert^2  \, dx \bigg] \, ds \\
& =\displaystyle \int_{0}^{t} \, 
\displaystyle \int_{\Omega_{\epsilon}}  \, 
\frac{\displaystyle \partial u_{\epsilon}^{(k)} (s)}{\displaystyle \partial s}
\, u_{\epsilon}^{(k)} (s) \, dx ds.
\end{split}
\end{equation}
Taking into account Eq. (\ref{2.3}) and Eq. (\ref{1.1}), we obtain
that for all $s \in [0, T_1]$:
\begin{equation} \label{2.9}
\begin{split}
&\displaystyle \int_{\Omega_{\epsilon}}
\frac{\displaystyle \partial u_{\epsilon}^{(k)} (s)}{\displaystyle \partial s}
\, u_{\epsilon}^{(k)} (s) \, dx=
\displaystyle \int_{\Omega_{\epsilon}}
\frac{\displaystyle \partial u_1^{\epsilon} (s)}{\displaystyle \partial s}
\, u_{\epsilon}^{(k)} (s) \, dx \\
&=\displaystyle \int_{\Omega_{\epsilon}}
\bigg[ d_1 \, \Delta_x u_1^{\epsilon}- u_1^{\epsilon} \,
\sum_{j=1}^{\infty} a_{1, j} u_j^{\epsilon}+ \sum_{j=1}^{\infty}
B_{1+j} \, \beta_{1+j,1} \, u_{1+j}^{\epsilon} 
\bigg] \, u_{\epsilon}^{(k)} (s) \, dx \\
&=\epsilon \, d_1 \displaystyle \int_{\Gamma_{\epsilon}}
\psi \bigg(s, x, \frac{\displaystyle x}{\displaystyle \epsilon} \bigg) \,
u_{\epsilon}^{(k)} (s) \, d\sigma_{\epsilon}(x)-
d_1 \displaystyle \int_{\Omega_{\epsilon}} \nabla_x u_1^{\epsilon} (s) \cdot
\nabla_x u_{\epsilon}^{(k)} (s) \, dx \\
&-\displaystyle \int_{\Omega_{\epsilon}} (u_1^{\epsilon}(s))^2 \, a_{1,1} \,
u_{\epsilon}^{(k)} (s) \, dx-
\displaystyle \int_{\Omega_{\epsilon}} u_1^{\epsilon}(s) \,
\sum_{j=2}^{\infty} \bigg[ a_{1,j} \, u_j^{\epsilon}(s) \bigg] \, 
u_{\epsilon}^{(k)} (s) \, dx \\
&+\displaystyle \int_{\Omega_{\epsilon}} \bigg[ \sum_{j=2}^{\infty} B_j \,
\beta_{j,1} \,  u_j^{\epsilon}(s) \bigg] \, u_{\epsilon}^{(k)} (s) \, dx \\
& \leq \epsilon \, d_1 \displaystyle \int_{\Gamma_{\epsilon}}
\psi \bigg(s, x, \frac{\displaystyle x}{\displaystyle \epsilon} \bigg) \,
u_{\epsilon}^{(k)} (s) \, d\sigma_{\epsilon}(x)-
d_1 \displaystyle \int_{\Omega_{\epsilon}} \nabla_x u_1^{\epsilon} (s) \cdot
\nabla_x u_{\epsilon}^{(k)} (s) \, dx \\
&-\displaystyle \int_{\Omega_{\epsilon}} \sum_{j=2}^{\infty}
\bigg[ a_{1,j} \, u_1^{\epsilon}(s)-B_j \, \beta_{j,1} \bigg] \,
u_j^{\epsilon} (s) \, u_{\epsilon}^{(k)} (s) \, dx.
\end{split}
\end{equation}
By using Assumption C, Lemma \ref{lA.0} and Young's inequality,
 one has, remembering that $k\ge \gamma_1$,
\begin{equation} \label{2.10}
\begin{split}
&\displaystyle \int_{\Omega_{\epsilon}}
\frac{\displaystyle \partial u_{\epsilon}^{(k)} (s)}{\displaystyle \partial s}
\, u_{\epsilon}^{(k)} (s) \, dx \leq
\dfrac{\epsilon \, d_1}{2} \displaystyle \int_{B_k^{\epsilon}(s)}
\bigg \vert \psi\bigg(s, x, \frac{\displaystyle x}{\displaystyle \epsilon}
\bigg) \bigg \vert^2 \,
 \, d\sigma_{\epsilon}(x) \\
&+ \dfrac{C_1 \, d_1}{2}
 \displaystyle \int_{A_k^{\epsilon}(s)}
|u_{\epsilon}^{(k)} (s) |^2 \, dx- 
d_1\left( 1-\dfrac{C_1 \epsilon^2}{2}\right) \displaystyle
\int_{\Omega_{\epsilon}} \vert \nabla_x u_{\epsilon}^{(k)} (s)\vert^2  \, dx ,
\end{split}
\end{equation}
where we denote by $A_{k}^{\epsilon} (t)$ and $B_{k}^{\epsilon} (t)$ the set
of points in $\Omega_{\epsilon}$ and on $\Gamma_{\epsilon}$, respectively,
at which $u_1^{\epsilon} (t, x) >k$.
We observe that
$$\vert A_{k}^{\epsilon} (t) \vert \leq \vert \Omega_{\epsilon} \vert, \qquad \vert B_{k}^{\epsilon} (t) \vert \leq \vert \Gamma_{\epsilon} \vert, $$
where $\mid \cdot \mid$ is the (resp. 3-dimensional and 2-dimensional) Lebesgue measure.
Inserting Eq. (\ref{2.10}) into Eq. (\ref{2.8}) and varying over $t$, 
we end up with the estimate:
\begin{equation} \label{2.11} 
\begin{split}
&\sup_{0 \leq t \leq T_1} \bigg[ 
\frac{\displaystyle 1}{\displaystyle 2} \displaystyle \int_{\Omega_{\epsilon}}
\vert u_{\epsilon}^{(k)} (t) \vert^2  \, dx \bigg]+d_1 \bigg( 1- 
\frac{C_1 \, \epsilon^2}{2} \bigg)
\displaystyle \int_{0}^{T_1} dt \, 
\displaystyle \int_{\Omega_{\epsilon}} 
\vert \nabla_x u_{\epsilon}^{(k)} (t) \vert^2 \, dx  \\ 
& \leq \frac{\displaystyle C_1 \, d_1}{\displaystyle 2} 
\displaystyle \int_{0}^{T_1}  dt \, 
\displaystyle \int_{A_{k}^{\epsilon} (t)}
\vert u_{\epsilon}^{(k)} (t) \vert^2  \, dx  
+\frac{\displaystyle \epsilon \, d_1}{\displaystyle 2} 
\displaystyle \int_{0}^{T_1} dt \,
\displaystyle \int_{B_{k}^{\epsilon} (t)} 
\bigg \vert \psi \bigg (t, x, \frac{x}{\epsilon} \bigg ) \bigg \vert^2 \,
d\sigma_{\epsilon}(x) .
\end{split}
\end{equation}
Introducing the norm (as in the Prop. above):
\begin{equation} \label{2.12}
\Vert u \Vert^{2}_{Q_{\epsilon} (T)} := \sup_{0 \leq t \leq T}
\displaystyle \int_{\Omega_{\epsilon}} \vert u(t) \vert^2 \, dx+
\displaystyle \int_{0}^{T} \, dt \displaystyle \int_{\Omega_{\epsilon}}
\vert \nabla u (t) \vert^2 \, dx,
\end{equation}
 inequality (\ref{2.11}) can be rewritten as follows:
\begin{equation} \label{2.13} 
\begin{split}
\min \bigg{\{\frac{\displaystyle 1}{\displaystyle 2},
d_1 \bigg( 1-\frac{\displaystyle C_1 \, \epsilon^2}{\displaystyle 2} \bigg) 
\bigg\}}
 \Vert u_{\epsilon}^{(k)} \Vert^{2}_{Q_{\epsilon} (T_1)} &\leq
\frac{\displaystyle C_1 \, d_1}{\displaystyle 2} \displaystyle \int_{0}^{T_1}  
dt \, 
\displaystyle \int_{A_{k}^{\epsilon} (t)}  
\vert u_{\epsilon}^{(k)} (t) \vert^2  \, dx \\
&+\frac{\displaystyle \epsilon \, d_1}{\displaystyle 2} 
\displaystyle \int_{0}^{T_1} dt
\displaystyle \int_{B_{k}^{\epsilon} (t)}
\bigg \vert \psi \bigg( t, x, \frac{\displaystyle x}{\displaystyle \epsilon}
\bigg) \bigg \vert^2 \,
d\sigma_{\epsilon}(x).
\end{split}
\end{equation}
Let us estimate the right-hand side of (\ref{2.13}).
From H\"older's inequality, we obtain
\begin{equation} \label{2.14}
\displaystyle \int_{0}^{T_1}  dt \,
\displaystyle \int_{A_{k}^{\epsilon} (t)}  \, 
\vert u_{\epsilon}^{(k)} (t) \vert^2 \, dx \leq
\Vert  u_{\epsilon}^{(k)} \Vert^{2}_{L^{\overline {r}_1} 
(0, T_1; L^{\overline {q}_1} (\Omega_{\epsilon}))} \,
\Vert \mathbb{1}_{A_{k}^{\epsilon}} \Vert_{L^{r'_1} (0, T_1; L^{q'_1} 
(\Omega_{\epsilon}))},
\end{equation}
with $r'_1=\frac{\displaystyle r_1}{\displaystyle r_1-1}$,
$q'_1=\frac{\displaystyle q_1}{\displaystyle q_1-1}$,
$\overline {r}_1 = 2 \, r_1$,
 $\overline {q}_1 = 2 \, q_1$, where
 $\overline {r}_1 \in
(2, \infty)$ and $\overline {q}_1 \in (2, 6)$
have been chosen in such a way that
$$\frac{\displaystyle 1}{\displaystyle \overline {r}_1}+
\frac{\displaystyle 3}{\displaystyle 2 \, \overline {q}_1}=
\frac{\displaystyle 3}{\displaystyle 4}.$$

In particular, $r'_1, q'_1<\infty$, so that \eqref{2.14} yields
\begin{equation} \label{2.15}
\displaystyle \int_{0}^{T_1}  dt \,
\displaystyle \int_{A_{k}^{\epsilon} (t)}  \, 
\vert u_{\epsilon}^{(k)} (t) \vert^2 \, dx \leq
\Vert  u_{\epsilon}^{(k)} \Vert^{2}_{L^{\overline {r}_1} 
(0, T_1; L^{\overline {q}_1} (\Omega_{\epsilon}))} \,
 \vert \Omega\vert^{1/q'_1} \, T_1^{1/r'_1}.
\end{equation}
If we choose (for $\epsilon>0$ small enough)
$$
T_1^{1/r'_1} < \dfrac{\min\{1,d_1\}}{2C_1d_1 c^2}\,  \vert \Omega\vert^{-1/q'_1} \le
 \dfrac{\min \bigg{\{\frac{\displaystyle 1}{\displaystyle 2},
d_1 \bigg( 1-\frac{\displaystyle C_1 \, \epsilon^2}{\displaystyle 2} \bigg) 
\bigg\}}}{C_1d_1 c^2 }\,  \vert \Omega\vert^{-1/q'_1} , 
$$
then from Lemma A.3 (i) (and $c$ being the constant appearing in formula (\ref{A.4}) of this Lemma) it follows that
\begin{equation} \label{2.16}
\frac{\displaystyle C_1 \, d_1}{\displaystyle 2}
\displaystyle \int_{0}^{T_1}  dt \displaystyle \int_{A_{k}^{\epsilon} (t)}
\vert u_{\epsilon}^{(k)} (t) \vert^2 \, dx \leq
\frac{\displaystyle 1}{\displaystyle 2} 
\min \bigg{\{\frac{\displaystyle 1}{\displaystyle 2},
d_1 \bigg( 1-\frac{\displaystyle C_1 \, \epsilon^2}{\displaystyle 2} \bigg) 
\bigg\}} \,
\Vert  u_{\epsilon}^{(k)} \Vert^{2}_{Q_{\epsilon} (T_1)}.
\end{equation}

Analogously, from H\"older's inequality, we have (remember that  $k \geq \hat{k}$)
\begin{equation} \label{2.17} 
\begin{split}
\frac{\displaystyle \epsilon \, d_1}{\displaystyle 2}
\displaystyle \int_{0}^{T_1}  dt \, \displaystyle \int_{B_{k}^{\epsilon} (t)}
\bigg \vert \psi \bigg( t, x, \frac{\displaystyle x}{\displaystyle \epsilon} 
\bigg) \bigg \vert^2 \, 
d\sigma_{\epsilon}(x) & \leq 
\frac{\displaystyle \epsilon \, d_1 \, k^2}{\displaystyle 2}
\bigg( \frac{\displaystyle \hat{k}^2}{\displaystyle k^2} \bigg) \,
\Vert \mathbb{1}_{B_{k}^{\epsilon}} \Vert_
{L^1 (0, T_1; L^1 (\Gamma_{\epsilon}))} \\
& \leq \frac{\displaystyle \epsilon \, d_1 \, k^2}{\displaystyle 2} \,
\displaystyle \int_{0}^{T_1}  dt \, \vert B_{k}^{\epsilon} (t) \vert.
\end{split}
\end{equation}
Thus, estimate (\ref{2.13}) yields
\begin{equation} \label{2.18} 
\Vert  u_{\epsilon}^{(k)} \Vert^{2}_{Q_{\epsilon} (T_1)} \leq
\epsilon \, \beta \, k^2 \,  
\displaystyle \int_{0}^{T_1}  dt \, 
\vert B_{k}^{\epsilon} (t) \vert ,
\end{equation}
with $\beta := \max(1, d_1) + 1/2$.
\par 
Hence, using Prop. \ref{t1.1} for $w^{\epsilon} := u_1^{\epsilon}$, we obtain
%\medskip
%{\bf{Passage below to check. Laurent}}
%\medskip
$$\Vert u_1^{\epsilon} \Vert_{L^{\infty} (0, T_1; L^{\infty} 
(\Gamma_{\epsilon}))} \leq C(\Omega, \beta, T_1)\, {\hat k},$$
where the positive constant $C(\Omega, \beta, T_1)$ does not depend on
%is independent of 
$\epsilon$ or $\hat{k}$.
\medskip

 The same argument can be repeated on the cylinder $[T_1, 2\,T_1]$ with 
 $k \ge \hat{k_1} := \max(\gamma_1, C(\Omega, \beta, T_1)\, {\hat k})$, 
  yielding 
 $$\Vert u_1^{\epsilon} \Vert_{L^{\infty} (0, 2\,T_1; L^{\infty} 
(\Gamma_{\epsilon}))} \leq C(\Omega, \beta, T_1)\, {\hat k_1}.$$
 Thanks to a straightforward induction, one gets the bound 
 for $u_1^{\epsilon}$ in $L^{\infty} (0, T; L^{\infty} 
(\Gamma_{\epsilon}))$.
%  are valid for the cylinder $[T_s, T_{s+1}] \times 
%\Omega_{\epsilon}$, $s=1, 2, \ldots, p-1$ with
%$$
%\bigg[ T_{s+1}- T_{s} \bigg]^{1/r'_1}
%< \dfrac{\min\{1,d_1\}}{2C_1d_1 c^2}\,  \vert \Omega\vert^{-1/q'_1}
%$$
%and $T_p := T$.
%Thus, after a finite number of steps, we obtain estimate (\ref{1.48}).

\end{proof}

%\begin{lemma} \label{l1.4}
%The sequence $\nabla_x u_1^{\epsilon}$ is bounded in
%$L^2([0,T] \times \Omega_{\epsilon})$, uniformly in $\epsilon$.
%\end{lemma}
%This Lemma can be easily proved by following the same arguments presented in
%\cite{10} (Lemma $5.4$), provided that the assumption (\ref{1.7a}) 
%is taken into account.
We finally write the following $L^{\infty}$ bound for all 
$u_i^{\epsilon}$:

\begin{lemma} \label{l1.5}
 Let $\Omega_{\epsilon}$ be an open set satisfying Assumption $0$.
We also suppose that Assmptions A, B, and C hold. We finally consider $T >0$, 
 and a classical 
solution $u_i^{\epsilon}$ ($i\in \N - \{0\}$) of (\ref{1.1}), (\ref{1.2}).
Then, the following uniform with respect to $\epsilon>0$ (small enough)
 estimate holds for all
$i\in \N - \{0\}$:
\begin{equation} \label{2.19}
\Vert u_i^{\epsilon} \Vert_{L^{\infty} (0, T; L^{\infty} (\Omega_{\epsilon}))}
\leq K_i,
\end{equation}
 where $K_1$ is given by
Lemma \ref{l1.2}, estimate (\ref{1.42}) and Lemma \ref{l1.3}, estimate (\ref{1.48}), and, for $i \ge 2$,
\begin{equation} \label{2.20}
K_i=1+\frac{\displaystyle \bigg[\sum_{j=1}^{i-1} a_{j,i-j} K_j K_{i-j}
\bigg]}{\displaystyle (B_i+a_{i,i})} + \gamma_i .
\end{equation}

\end{lemma}

\begin{proof}

The Lemma can be proved directly by induction following the proof reported in
\cite{22} (Lemma $2.2$, p. 284). 
Since we have a zero initial condition for the system (\ref{1.2}), we have
chosen a function slightly different from the one used
 in \cite{22} to test
the $i$-th equation of (\ref{1.2}), namely
$$\phi_i := p \, (u_i^{\epsilon})^{(p-1)} \; \; \; \; p \geq 2.$$
We stress that the functions $\phi_i$ are strictly positive  and continuously
differentiable on $[0,t] \times \overline\Omega$, for all $t>0$.

Therefore, multiplying the $i$-th equation in system (\ref{1.2}) by $\phi_i$ and reorganizing the
terms appearing in the sums, we can write the estimate
$$ ||u_i^{\epsilon}||_{L^p(\Omega_\epsilon)}^p + d_i\,p\,(p-1) \int_0^t \int_{\Omega_\epsilon}  
|\nabla_x u_i^{\epsilon}|^2 \, (u_i^{\epsilon})^{p-2} dx ds  $$
$$ \le  \int_0^t \int_{\Omega_\epsilon}  \bigg[ \frac12 \sum_{j=1}^{i-1} a_{i-j,j} \, u_j^{\epsilon}\,  u_{i-j}^{\epsilon}
- a_{i,i}\, | u_i^{\epsilon}|^2 - B_i  \,u_i^{\epsilon} \bigg] \, p\, (u_i^{\epsilon})^{p-1} dx ds  $$
$$ -  \int_0^t \int_{\Omega_\epsilon}  \bigg[  \sum_{j=1}^{i-1} a_{i,j} \, u_i^{\epsilon}\,  u_{j}^{\epsilon}
+ \sum_{j=i+1}^{\infty}  (a_{i,j}\,  u_i^{\epsilon} - B_j \, \beta_{j,i} )\, \,u_j^{\epsilon} \bigg] \, p\, (u_i^{\epsilon})^{p-1} dx ds . $$
We now work using an induction on $i$. 
Supposing that we already know that $\Vert u_j^{\epsilon} \Vert_{L^{\infty} (0, T; L^{\infty} (\Omega_{\epsilon}))}
\leq K_j $ for all $j<i$, and using assumption C,
%(\ref{1.7a}), 
the previous estimate leads to
$$ ||u_i^{\epsilon}||_{L^p(\Omega_\epsilon)}^p  \le \int_0^t \int_{\Omega_\epsilon}  
\bigg[ \frac12 \sum_{j=1}^{i-1} a_{i-j,j} \, K_j\,  K_{i-j}
- a_{i,i}\, | u_i^{\epsilon}|^2 - B_i  \,u_i^{\epsilon} \bigg] \, p\, (u_i^{\epsilon})^{p-1} dx ds  $$
$$ +  \int_0^t \int_{\Omega_\epsilon}   \sum_{j=i+1}^{\infty}  a_{i,j}\,  (-u_i^{\epsilon} + \gamma_i) \, \,u_j^{\epsilon} \, p\, (u_i^{\epsilon})^{p-1} dx ds =: I_1 + I_2. $$
Then, optimizing w.r.t. $u_i^{\epsilon}$,
%thanks for example to Young's inequality,
$$ I_1 \le \bigg[ \bigg(\sum_{j=1}^{i-1} a_{i-j,j} \, K_j\,  K_{i-j} \bigg)^p\, (B_i + a_{i,i})^{1-p} \bigg]\, |\Omega_{\epsilon} |\, T
+ p\, a_{i,i}\, |\Omega_{\epsilon} |\, T, $$
and 
$$ I_2 \le  \int_0^t \int_{\Omega_\epsilon}  \sum_{j=i+1}^{\infty}  a_{i,j}\, (\gamma_i - u_i^{\epsilon}) \, u_j^{\epsilon}\, 1_{\{u_i^\epsilon \le \gamma_i\}}
 \, p\, (u_i^{\epsilon})^{p-1} dx ds $$
 $$ \le p\, \gamma_i^p  \int_0^t \int_{\Omega_\epsilon} \bigg( \sum_{j=i+1}^{\infty}  a_{i,j} \,  u_j^{\epsilon} \bigg) dx ds $$
$$ \le C\, p\, \gamma_i^p  ( |\Omega_{\epsilon} |\, T )^{1/2}, $$
where Cauchy-Schwarz inequality and the duality Lemma  
(more precisely Eq. (\ref{1.41})) have been exploited.
\medskip

Using these estimates for bounding $ ||u_i^{\epsilon}||_{L^p(\Omega_\epsilon)}$ and letting $p \to \infty$, we end up with the desired
estimate.

\end{proof}

%\begin{lemma} \label{l1.6}
%The sequence $\nabla_x u_i^{\epsilon}$ ($ i \geq 2$) is bounded in
%$L^2([0,T] \times \Omega_{\epsilon})$, uniformly in $\epsilon$.
%\end{lemma}
%This Lemma can be easily proved by following the same arguments presented in
%\cite{10} (Lemma $5.6$), provided that the assumption (\ref{1.7a})
%is taken into account.

We end up this section with bounds for the derivatives of
$u_i^{\epsilon}$.

\begin{lemma} \label{l1.7}
Let $\Omega_{\epsilon}$ be an open set satisfying Assumption $0$.
We also suppose that Assmptions A, B, and C hold. We finally consider $T >0$,  and a classical 
solution $u_i^{\epsilon}$ ($i\in \N - \{0\}$, $\epsilon>0$ small enough) of (\ref{1.1}), (\ref{1.2}).
Then, the family $\partial_t u_i^{\epsilon}$ 
is bounded in
$L^2([0,T] \times \Omega_{\epsilon})$, and the
 family $\nabla_x u_i^{\epsilon}$ 
is bounded in
$L^{\infty}([0,T]; L^2( \Omega_{\epsilon}))$,
 uniformly in $\epsilon$ (but not in $i$).
\end{lemma}

\begin{proof}
Since this proof is close to the proof of Lemma $5.9$ in \cite{10}, we only
sketch it.

Case $i=1$: Let us multiply the first equation in (\ref{1.1}) by the function
$\partial_t u_1^{\epsilon} (t, x)$.
Integrating, the divergence theorem yields

\begin{equation} \label{2.21} 
\begin{split}
&\displaystyle \int_{\Omega_{\epsilon}} \bigg \vert 
\frac{\partial u_1^{\epsilon} (t, x)}{\partial t} \bigg \vert^2 \, dx +
\frac{\displaystyle d_1}{\displaystyle 2} 
\displaystyle \int_{\Omega_{\epsilon}} 
\frac{\partial}{\partial t} (\vert \nabla_x u_1^{\epsilon} (t,x) \vert^2) \, dx
\\
&=\epsilon \, d_1 \, \displaystyle \int_{\Gamma_{\epsilon}} 
\psi \bigg(t,x,\frac{x}{\epsilon} \bigg) \, 
\frac{\partial u_1^{\epsilon}}{\partial t} \, d\sigma_{\epsilon}(x)-
\displaystyle \int_{\Omega_{\epsilon}} \, u_1^{\epsilon} \,
\bigg(\sum_{j=1}^{\infty} a_{1,j} \, u_j^{\epsilon} \bigg) \,
\frac{\partial u_1^{\epsilon}}{\partial t} \, dx \\
&+\displaystyle \int_{\Omega_{\epsilon}} \bigg(\sum_{j=1}^{\infty}
B_{1+j} \, \beta_{1+j,1} \, u_{1+j}^{\epsilon} \bigg) \, 
\, \frac{\partial u_1^{\epsilon}}{\partial t} \, dx.
\end{split}
\end{equation}
Using Young's inequality and exploiting the boundedness of 
$u_1^{\epsilon}$ in $L^{\infty} (0, T; L^{\infty} (\Omega_{\epsilon}))$,
one gets 
\begin{equation} \label{2.22}
\begin{split}
&C_1 \, \displaystyle \int_{\Omega_{\epsilon}} \bigg \vert
\frac{\partial u_1^{\epsilon} (t, x)}{\partial t} \bigg \vert^2 \, dx +
\frac{\displaystyle d_1}{\displaystyle 2} 
\displaystyle \int_{\Omega_{\epsilon}}
\frac{\partial}{\partial t} (\vert \nabla_x u_1^{\epsilon} (t,x) \vert^2) \, dx
\\
& \leq \epsilon \, d_1 \, \displaystyle \int_{\Gamma_{\epsilon}}
\psi \bigg(t,x,\frac{x}{\epsilon} \bigg) \,
\frac{\partial u_1^{\epsilon}}{\partial t} \, d\sigma_{\epsilon}(x)+
C_2 \, \displaystyle \int_{\Omega_{\epsilon}} \, \bigg \vert
\sum_{j=1}^{\infty} a_{1,j} \, u_j^{\epsilon} \bigg \vert^2 \, dx \\
&+C_3 \, \displaystyle \int_{\Omega_{\epsilon}} \, \bigg \vert
\sum_{j=2}^{\infty} B_j \, \beta_{j,1} \, u_j^{\epsilon} \bigg \vert^2 \, dx,
\end{split}
\end{equation}
where $C_1$, $C_2$ and $C_3$ are positive constants which do not depend on
$\epsilon$.
Integrating over $[0, t]$ with $t \in [0, T]$, thanks to estimate (\ref{1.41}) and
Assumption C, we end up with the estimate
\begin{equation} \label{2.23}
\begin{split}
C_1 \, \displaystyle \int_{0}^{t} ds \, \displaystyle \int_{\Omega_{\epsilon}}
\bigg \vert \frac{\partial u_1^{\epsilon}}{\partial s} \bigg \vert^2 \, dx &+
\frac {\displaystyle d_1}{\displaystyle 2} 
\displaystyle \int_{\Omega_{\epsilon}} 
\vert \nabla_x u_1^{\epsilon} (t, x) \vert^2 \, dx \leq C_4 \\ 
&+ \, \epsilon \, d_1 \displaystyle \int_{\Gamma_{\epsilon}}
\psi \bigg(t,x,\frac{x}{\epsilon} \bigg) \, u_1^{\epsilon} (t, x) \, 
d\sigma_{\epsilon}(x) \\
&- \, \epsilon \, d_1 \displaystyle \int_{0}^{t} ds \,
\displaystyle \int_{\Gamma_{\epsilon}} 
\frac{\partial}{\partial s} \psi \bigg(s,x,\frac{x}{\epsilon} \bigg) \,
u_1^{\epsilon} (s, x) \, d\sigma_{\epsilon}(x),
\end{split}
\end{equation}
since $\psi \bigg(t=0,x,\frac{\displaystyle x}{\displaystyle \epsilon} \bigg) 
\equiv 0$.

Applying once more Young's inequality and taking into account estimate
(\ref{1.29}) and Lemma \ref{lA.0}, estimate (\ref{2.23}) can be rewritten as
follows
\begin{equation} \label{2.24}
C_1 \, \displaystyle \int_{0}^{t} ds \, \displaystyle \int_{\Omega_{\epsilon}}
\bigg \vert \frac{\partial u_1^{\epsilon}}{\partial s} \bigg \vert^2 \, dx+
\frac {\displaystyle d_1}{\displaystyle 2} (1-{\epsilon}^2 \, C_5) \,
\displaystyle \int_{\Omega_{\epsilon}} 
\vert \nabla_x u_1^{\epsilon} (t, x) \vert^2 \, dx 
%\leq C_6,
\end{equation}
$$ \leq C_6 + C_1\, \frac{d_1}2 \, \epsilon^2 \,\int_0^t
 \int_{\Omega_{\epsilon}} 
\vert \nabla_x u_1^{\epsilon} (s, x) \vert^2 \, dx ds, $$ 
where the positive constants $C_1$, $C_5$, $C_6$ do not depend on
$\epsilon$, since $\psi \in L^{\infty} (0, T; B)$, $u_1^{\epsilon}$ is
bounded in $L^{\infty} (0, T; L^{\infty} (\Omega_{\epsilon}))$,
%$ \nabla_x u_1^{\epsilon}$ is bounded in $L^2 (0, T; L^2 (\Omega_{\epsilon}))$
and the following inequality holds: 
\begin{equation} \label{2.25}
\epsilon \displaystyle \int_{\Gamma_{\epsilon}}
\bigg \vert \partial_t \psi \bigg(t,x,\frac{x}{\epsilon} \bigg) \bigg \vert^2
\, d\sigma_{\epsilon}(x) \leq C_7 \, \Vert \partial_t \psi (t)
\Vert_B^2 \leq C_8,
\end{equation}
with $C_7$ and $C_8$ which do not depend on $\epsilon$.
%For a sequence $\epsilon$ of positive numbers going to zero: 
%$(1-\epsilon^2 C_5) \geq 0$. Then, the second term on the left-hand side of 
%(\ref{2.24}) is nonnegative, and one has:
Then, using Gronwall's lemma,
\begin{equation} \label{2.26}
\Vert \partial_t u_1^{\epsilon} \Vert^2_{L^2 (0, T; L^2 (\Omega_{\epsilon}))}
\leq C,
\end{equation}
and
\begin{equation} \label{2.26bis}
\Vert \nabla_x u_1^{\epsilon} \Vert^2_{L^{\infty} (0, T; L^2 (\Omega_{\epsilon}))}
\leq C,
\end{equation}
where $C \geq 0$ is a constant which does not depend on $\epsilon$.
\medskip

Case $i \geq 2$: Let us multiply the first equation in (\ref{1.2}) by the 
function
$\partial_t u_i^{\epsilon} (t, x)$.
Integrating, the divergence theorem yields

\begin{equation} \label{2.27} 
\begin{split}
&\displaystyle \int_{\Omega_{\epsilon}} \bigg \vert 
\frac{\partial u_i^{\epsilon} (t, x)}{\partial t} \bigg \vert^2 \, dx +
\frac{\displaystyle d_i}{\displaystyle 2} 
\displaystyle \int_{\Omega_{\epsilon}} 
\frac{\partial}{\partial t} (\vert \nabla_x u_i^{\epsilon} (t,x) \vert^2) \, 
dx \\
&=\frac{\displaystyle 1}{\displaystyle 2} \, 
\displaystyle \int_{\Omega_{\epsilon}} 
\bigg( \sum_{j=1}^{i-1} \, a_{i-j,j} \, u_{i-j}^{\epsilon} \,
u_j^{\epsilon} \bigg) \, \frac{\partial u_i^{\epsilon}}{\partial t} \, dx-
\displaystyle \int_{\Omega_{\epsilon}} \, u_i^{\epsilon} \,
\bigg(\sum_{j=1}^{\infty} a_{i,j} \, u_j^{\epsilon} \bigg) \,
\frac{\partial u_i^{\epsilon}}{\partial t} \, dx \\
&+\displaystyle \int_{\Omega_{\epsilon}} \bigg(\sum_{j=1}^{\infty}
B_{i+j} \, \beta_{i+j,i} \, u_{i+j}^{\epsilon} \bigg) \, 
\, \frac{\partial u_i^{\epsilon}}{\partial t} \, dx-
\displaystyle \int_{\Omega_{\epsilon}} \, B_i \, u_i^{\epsilon} \,
\frac{\partial u_i^{\epsilon}}{\partial t} \, dx.
\end{split}
\end{equation}

Using Young's inequality and exploiting the boundedness of
$u_i^{\epsilon}$ in $L^{\infty} (0, T; L^{\infty} (\Omega_{\epsilon}))$,
one gets
\begin{equation} \label{2.28}
\begin{split}
&C_1 \, \displaystyle \int_{\Omega_{\epsilon}} \bigg \vert
\frac{\partial u_i^{\epsilon} (t, x)}{\partial t} \bigg \vert^2 \, dx +
\frac{\displaystyle d_i}{\displaystyle 2} 
\displaystyle \int_{\Omega_{\epsilon}}
\frac{\partial}{\partial t} (\vert \nabla_x u_i^{\epsilon} (t,x) \vert^2) \, dx
\\
& \leq C_2+
C_3 \, \displaystyle \int_{\Omega_{\epsilon}} \, \bigg \vert
\sum_{j=1}^{\infty} a_{i,j} \, u_j^{\epsilon} \bigg \vert^2 \, dx 
+C_4 \, \displaystyle \int_{\Omega_{\epsilon}} \, \bigg \vert
\sum_{j=i+1}^{\infty} B_j \, \beta_{j,i} \, u_j^{\epsilon} \bigg \vert^2 \, dx,
\end{split}
\end{equation}
where $C_1$, $C_2$, $C_3$ and $C_4$ are positive constants 
which do not depend on
$\epsilon$.

Integrating over $[0, t]$ with $t \in [0, T]$, thanks to estimate (\ref{1.41}) and
Assumption C, we end up with the estimate
\begin{equation} \label{2.29}
C_1 \, \displaystyle \int_{0}^{t} ds \, \displaystyle \int_{\Omega_{\epsilon}}
\bigg \vert \frac{\partial u_i^{\epsilon}}{\partial s} \bigg \vert^2 \, dx +
\frac {\displaystyle d_i}{\displaystyle 2} 
\displaystyle \int_{\Omega_{\epsilon}} 
\vert \nabla_x u_i^{\epsilon} (t, x) \vert^2 \, dx \leq C_5  ,
\end{equation}
with $C_5 \geq 0$ independent of $\epsilon$ (but not on $i$).
%Since the second term on the left-hand side of (\ref{2.29}) is nonnegative,
We conclude that
\begin{equation} \label{2.30}
\Vert \partial_t u_i^{\epsilon} \Vert^2_{L^2 (0, T; L^2 (\Omega_{\epsilon}))}
\leq C,
\end{equation}
and
\begin{equation} \label{2.30bis}
\Vert \nabla_x u_i^{\epsilon} \Vert^2_{L^{\infty} (0, T; L^2 (\Omega_{\epsilon}))}
\leq C,
\end{equation}
where $C \geq 0$ is a constant independent of $\epsilon$ (but not on $i$).

\end{proof}

This concludes the section devoted to {\it{a priori}} estimates which are uniform w.r.t. the homogenization parameter $\epsilon$.

\section{Proof of the main result} \label{section3}

We start here the proof of our main Theorem \ref{t2.1}.

\subsection{Existence of solutions for a given $\epsilon >0$} \label{subsec31}

We first explain how
% the estimates of the previous section can be used in the 
to get a proof of existence, for a given $\var>0$, of 
a (strong) solution to system \eqref{1.1} - \eqref{1.2}. 
We state the:
\medskip 

\begin{proposition} \label{ppr}
 Let $\epsilon>0$ small enough be given, $\Omega_{\epsilon}$ be a bounded regular open set of $\R^3$, and consider data satisfying
Assumptions A, B and C. Then there exists a
solution $(u_i^{\epsilon})_{i \ge 1}$ to system \eqref{1.1} - \eqref{1.2}, which is strong in the following sense: For all $T>0$ and $i\ge 1$, $u_i^{\epsilon} \in L^{\infty}([0,T] \times \Omega_{\epsilon})$, $\frac{\partial u_i^{\epsilon}}{\partial t} \in L^2([0,T] \times \Omega_{\epsilon})$, $\frac{\partial^2 u_i^{\epsilon}}{\partial x_k \partial x_l} \in L^2([0,T] \times \Omega_{\epsilon})$ for all $k,l \in \{1,..,3\}$.
\end{proposition}

\begin{proof}
We introduce a finite size truncation of this system,
which writes (once the notation of the dependence w.r.t. $\var$ of the unknowns has been eliminated 
for readability):
\begin{eqnarray} \label{1.1tr}
\begin{cases}
\frac{\displaystyle \partial{u_1^{n}}}{\displaystyle \partial t}-
d_1\, \Delta_x  u_1^{n} + u_1^{n} \, \sum_{j=1}^{n} 
a_{1,j}
u_j^{n}=\sum_{j=1}^{n-1} B_{1+j} \, \beta_{1+j,1} \, 
u_{1+j}^{n}
& \text{in } [0,T] \times \Omega_{\epsilon}, \\

\\ 
\frac{\displaystyle \partial{u_1^{n}}}{\displaystyle \partial \nu} 
:= \nabla_x u_1^{n}
\cdot n=0 & \text{on } [0,T] \times \partial\Omega , \\

\\
\frac{\displaystyle \partial{u_1^{n}}}{\displaystyle \partial \nu}
:= \nabla_x u_1^{n} \cdot n=
\epsilon \, \psi(t, x, \frac{x}{\epsilon}) & \text{on } 
[0,T] \times \Gamma_{\epsilon} , \\

\\
u_1^{n}(0,x)=U_1 & \text{in } \Omega_{\epsilon} ,
\end{cases}
\end{eqnarray}

and, if $i =2,..,n$,

\begin{eqnarray} \label{1.2tr}
\begin{cases}
\frac{\displaystyle \partial{u_i^{n}}}{\displaystyle \partial t}-
 d_i \,\Delta_x u_i^{n}
=Q_i^{n}+F_i^{n} & \text{in } [0,T] \times \Omega_{\epsilon} ,\\

\\ 
\frac{\displaystyle \partial{u_i^{n}}}{\displaystyle \partial \nu} 
:= \nabla_x u_i^{n}
\cdot n=0 & \text{on } [0,T] \times \partial\Omega , \\

\\
\frac{\displaystyle \partial{u_i^{n}}}{\displaystyle \partial \nu}
:= \nabla_x u_i^{n} \cdot n=0
& \text{on } 
[0,T] \times \Gamma_{\epsilon} , \\

\\
u_i^{n}(0,x)=0  & \text{in } \Omega_{\epsilon} , 
\end{cases}
\end{eqnarray}
where the truncated coagulation and breakup kernels $Q_i^{n}$, $F_i^{n}$ write

\begin{equation} \label{1.3tr}
Q_i^{n}:=\frac{1}{2} \, \sum_{j=1}^{i-1} a_{i-j,j} \, 
u_{i-j}^{n} \, u_j^{n}-\sum_{j=1}^{n} a_{i,j} \,
u_i^{n} \, u_j^{n},
\end{equation}

\begin{equation} \label{1.4tr}
F_i^{n}:=\sum_{j=1}^{n-i} B_{i+j} \, 
\beta_{i+j,i} \, u_{i+j}^{n}-B_i \,
u_i^{n}.
\end{equation}

We then observe  that the duality lemma (that is, Lemma \ref{l1.1} and Corollary \ref{cor22}) 
is still valid in this setting (with a proof that exactly follows 
the proof written above), so that we end up with the {\it{a priori}} estimate 
\begin{equation} \label{stre} \int_0^T\int_{\Omega_\var} \bigg| \sum_{i=1}^n i\, u_i^n(t,x) \bigg|^2 \, dt dx \le C, 
\end{equation}
where $C$ is a constant which does not depend on $n$.
% (but which may depend on $i$). 
\medskip

Using now a proof analogous to that of Lemmas \ref{l1.2} to \ref{l1.5}, we can obtain the 
{\it{a priori}} estimate 
\begin{equation}\label{newe9}
||u_i^n||_{L^{\infty}([0,T] \times \Omega_{\var})} \le C_i, 
\end{equation} 
where $C_i>0$ is a constant which also does not depend on $n$ (but may depend on $i$)
%in fact we will not 
%use the uniformity w.r.t. $n$ of this bound in the sequel).
\medskip

At this point, we use  standard theorems for systems of reaction-diffusion 
equations in order to get the existence and uniqueness
of a smooth solution to system  \eqref{1.1tr} - \eqref{1.2tr} 
(for a given $n \in \N - \{0\}$). 
We refer to \cite{LD_parma}, Prop. 3.2 p. 97 and Thm. 3.3 p. 105 for a 
complete description of a case with a slightly different boundary condition 
(homogeneous Neumann instead of
inhomogeneous Neumann) and a different right-hand side (but having the same 
crucial property, that is leading to an $L^{\infty}$ {\it{a priori}} bound on the 
components of the unknown). 
\medskip

We now briefly explain how to pass to the limit when $n \to \infty$ in such a 
way that the limit of $u_i^n$ satisfies the
system \eqref{1.1} - \eqref{1.2}. 
First, we notice that thanks to the duality estimate (\ref{stre}), 
each component sequence $(u_i^n)_{n \ge i}$ is bounded in
$L^2([0,T] \times \Omega_{\var})$. 
As a consequence, we can extract a subsequence from $(u_i^n)_{n \ge i}$ still 
denoted by
 $(u_i^n)_{n \ge i}$ (the extraction is done diagonally in such a way that it 
gives a subsequence which is common for all $i$)
 which converges in $L^2([0,T] \times \Omega_{\var})$ weakly towards some 
function $u_i \in L^2([0,T] \times \Omega_{\var})$. 
 Using then the {\it{a priori}} estimates (\ref{newe9}) and (\ref{stre}),
 we see that
 % (consequence of the duality lemma, the 
%assumptions on the coagulation and fragmentation coefficients, and 
%  natural bound on $u_i^n$ in $L^1$ coming from a direct integration of the 
%equations)
\begin{equation}\label{newe10}
  ||\frac{\displaystyle \partial{u_i^{n}}}{\displaystyle \partial t}-
d_i \Delta_x u_i^{n} ||_{L^2([0,T] \times \Omega_{\var})} 
\le C_i, 
\end{equation}
where $C_i$ may depend on $i$ but not on $n$, so that the convergence 
in fact holds for a.e $(t,x) \in [0,T] \times \Omega_{\var}$.
This is sufficient to pass to the limit in system 
\eqref{1.1tr} - \eqref{1.2tr}  and get a weak solution $(u_i)_{i \ge 1}$ to system \eqref{1.1} - \eqref{1.2}. Moreover, thanks to estimates (\ref{newe9}) and (\ref{newe10}), this solution 
is strong, in the sense that for all $T>0$, $u_i^{\epsilon} \in L^{\infty}([0,T] \times \Omega_{\epsilon})$, $\frac{\partial u_i^{\epsilon}}{\partial t} \in L^2([0,T] \times \Omega_{\epsilon})$, $\frac{\partial^2 u_i^{\epsilon}}{\partial x_k \partial x_l} \in L^2([0,T] \times \Omega_{\epsilon})$ for all $k,l \in \{1,..,3\}$.
\end{proof}

\subsection{Homogenization} \label{subsec32}

We now present the end of the proof of our main Theorem \ref{t2.1}, in which we use the solutions  to system \eqref{1.1} - \eqref{1.2} for a given $\epsilon>0$ obtained in Prop. \ref{ppr}, and the (uniform w.r.t. $\epsilon$) {\it{a priori}} estimates of Section 2, in order to perform the homogenization process corresponding to the limit
$\epsilon \to 0$.
\medskip

We recall that we use the notation $\,\widetilde{}\,$ for the extension by $0$ to $\Omega$ of
functions defined on $\Omega_{\epsilon}$, and the notation $\chi$ for the characteristic
function of $Y^*$.

%Now that existence for a given $\epsilon$ is obtained, we provide the proof of 
%the homogenization part of Theorem \ref{t2.1}.
\medskip

In view of Lemmas \ref{l1.5} and \ref{l1.7},
the sequences $\widetilde{u_i^{\epsilon}}$,  
$\widetilde{\nabla_x u_i^{\epsilon}}$ and $ \widetilde{\frac{\displaystyle \partial u_i^{\epsilon}}
{\displaystyle \partial t}} $ ($i \geq 1$) are bounded in 
$L^2 ([0,T] \times \Omega)$.  Using Proposition \ref{tB.1} and 
Proposition \ref{tB.3}, and following \cite{3}, Thm 2.9, p.1498 which is specially designed for perforated domains (in the elliptic case, but the 
transfer to the parabolic case is easy)
%\medskip
%
%{\bf{We should precise the theorems above or add some explanations/references to explain the form of the limits below (which is specific to perforated domains); $\chi$ should be introduced somewhere in those theorems)}}
%\medskip
%
 they two-scale converge, up to a subsequence, respectively, to functions of the form:
$[(t,x,y)\mapsto \chi(y) \, u_i(t,x)]$, 
$[(t,x,y) \mapsto \chi(y) \, (\nabla_x u_i(t,x)+\nabla_y u_i^1(t,x,y))]$,
and $\bigg[ (t,x,y) \to \chi(y) \, \frac{\displaystyle \partial u_i}
{\displaystyle \partial t} (t,x) \bigg]$, for $i \geq 1$.
\par 
In the formulas above, $u_i \in L^2 (0,T; H^1 (\Omega))$ and 
$ u_i^1 \in L^2 ([0,T] \times \Omega; H_{\#}^1(Y)/\mathbb{R})$. 
%(cf. \cite{3}).
% Similarly, in view of Lemma \ref{l1.7}, it is possible 
%to prove that the
%sequence $\bigg( \widetilde{\frac{\displaystyle \partial u_i^{\epsilon}}
%{\displaystyle \partial t}} \bigg)$ ($i \geq 1$) two-scale
%converges to: $\bigg[ \chi(y) \, \frac{\displaystyle \partial u_i}
%{\displaystyle \partial t} (t,x) \bigg]$ ($i \geq 1$). 
%We can now find the homogenized equations satisfied by $u_i (t,x) $ and
%$u_i^1 (t,x,y)$.

In the case when $i=1$, let us multiply the first equation of (\ref{1.1}) 
by the test function $(t,x) \mapsto \phi_{\epsilon}(t,x,\frac{x}{\epsilon})$,
where
\begin{equation}\label{phiep}
\phi_{\epsilon}(t,x,y) := \phi(t,x)+\epsilon \, \phi_1(t,x,y), 
\end{equation}
with $\phi \in C^1 ([0,T] \times \overline {\Omega})$, and
$\phi_1 \in C^1 ([0,T] \times \overline {\Omega}; C_{\#}^{\infty}(Y))$.
Integrating, the divergence theorem yields
\begin{equation} \label{4.4} 
\begin{split}
&\displaystyle \int_0^T \int_{\Omega_{\epsilon}} 
\frac{\displaystyle \partial u_1^{\epsilon}}{\displaystyle \partial t} \,
\phi_{\epsilon} (t,x,\frac{x}{\epsilon}) \, dt \, dx+
d_1 \, \displaystyle \int_0^T \int_{\Omega_{\epsilon}} \nabla_x u_1^{\epsilon} 
\cdot  \nabla_x \left[(t,x) \mapsto \phi_{\epsilon} (t,x,\frac{x}{\epsilon}) \right] \, dt \, dx \\ 
&+\displaystyle \int_0^T \int_{\Omega_{\epsilon}} u_1^{\epsilon}
\sum_{j=1}^{\infty} a_{1,j} \, u_j^{\epsilon} \, \phi_{\epsilon}(t,x,\frac{x}{\epsilon}) \, dt \, dx=
\epsilon \, d_1 \, \displaystyle \int_0^T \int_{\Gamma_{\epsilon}}
\psi \bigg( t,x,\frac{x}{\epsilon} \bigg) \, \phi_{\epsilon}(t,x,\frac{x}{\epsilon}) \, dt \,
d\sigma_{\epsilon}(x) \\
&+\displaystyle \int_0^T \int_{\Omega_{\epsilon}}
\sum_{j=1}^{\infty} B_{1+j} \, \beta_{1+j,1} \, u_{1+j}^{\epsilon} \,
\phi_{\epsilon}(t,x,\frac{x}{\epsilon}) \, dt \, dx .
\end{split}
\end{equation}
Using the two-scales convergences described above, we can directly pass 
to the limit in the two first terms of this weak formulation. It is also easy to pass to the limit in the fourth one thanks to Prop. B.6.
\medskip

%*** We should give a reference for the identification of the limit of the boundary term
%(Thm B.5, as it is, just enables to show that a subsequence converges). Laurent ***

%\medskip
%Passing to the two-scale limit,
% thanks to Theorem \ref{tB.2} and 
%Theorem \ref{tB.5},

 The passage to the limit in the last infinite sum can be performed  thanks 
to Assumption C, 
%on $a_{i,j}$, $B_j$ and $\beta_{i,j}$, 
 the duality Lemma \ref{l1.1} (estimate (\ref{1.8})),
and Cauchy-Schwarz inequality (used in the last inequality below), indeed
$$ \bigg| \int_0^T \int_{\Omega_{\epsilon}}
\sum_{j=K}^{\infty} B_{1+j} \, \beta_{1+j,1} \, u_{1+j}^{\epsilon} \,
\phi_{\epsilon} \, dt \, dx \bigg| $$
$$ \le \int_0^T \int_{\Omega_{\epsilon}}
\sum_{j=K}^{\infty} \gamma_1\, a_{1,1+j}\, \, u_{1+j}^{\epsilon} \,\, dt \, dx \, \, ||\phi_{\epsilon}||_{\infty}$$
$$ \le C\, \int_0^T \int_{\Omega_{\epsilon}}
\sum_{j=K}^{\infty} (1+j)^{1 - \zeta}   \, u_{1+j}^{\epsilon} \,\, dt \, dx $$
$$ \le C \, K^{-\zeta} , $$
where $C$ does not depend on $\epsilon$.
\medskip

 The infinite sum in the third term of identity (\ref{4.4}) can be treated in the same way, using moreover Lemma \ref{l1.5}, indeed  
$$ \bigg| \int_0^T \int_{\Omega_{\epsilon}}
u_1^{\epsilon} \,\sum_{j=K}^{\infty}  a_{1,j} \, u_{j}^{\epsilon} \,
\phi_{\epsilon} \, dt \, dx \bigg| $$
$$ \le C\,\int_0^T \int_{\Omega_{\epsilon}}
\sum_{j=K}^{\infty} \, a_{1,j}\, \, u_{j}^{\epsilon} \,\, dt \, dx $$
$$ \le C \, K^{-\zeta} , $$
where $C$ does not depend on $\epsilon$.
\medskip

Note that the passage to the limit in quadratic terms like $u_1^{\epsilon}\, u_{j}^{\epsilon}$ can be performed thanks to Prop. B.3 (and the remark after this proposition), as done in \cite{10}.
%*** We should include some justification/reference for the strong passage to the limit enabling to treat the (finite sum of) products, based on the remark after Prop. B.3. Laurent ***
\medskip 

Finally, the passage to the limit leads to the variational formulation:
\begin{equation} \label{4.5} 
\begin{split}
&\displaystyle \int_0^T \int_{\Omega} \int_{Y^{\ast}}
\frac{\displaystyle \partial u_1}{\displaystyle \partial t}(t,x) \, 
\phi(t,x) \, dt \, dx \, dy \\
&+d_1 \, \displaystyle \int_0^T \int_{\Omega} \int_{Y^{\ast}}
[\nabla_x u_1(t,x)+\nabla_y u_1^1(t,x,y)] \cdot
[\nabla_x \phi(t,x)+\nabla_y \phi_1(t,x,y)] \, dt \, dx \, dy \\ 
&+\displaystyle \int_0^T \int_{\Omega} \int_{Y^{\ast}} u_1(t,x)
\sum_{j=1}^{\infty} a_{1,j} \, u_j(t,x) \, \phi(t,x) \, dt \, dx \, dy \\
&=d_1 \, \displaystyle \int_0^T \int_{\Omega} \int_{\Gamma} \psi(t,x,y) \,
\phi(t,x) \, dt \, dx \, d\sigma(y) \\
&+\displaystyle \int_0^T \int_{\Omega} \int_{Y^{\ast}}
\sum_{j=1}^{\infty} B_{1+j} \, \beta_{1+j,1} \, u_{1+j}(t,x) \, \phi(t,x) \, 
dt \,  dx \, dy.
\end{split}
\end{equation}

Thanks to an integration by parts, we see that (\ref{4.5}) can be put
in the strong form 
%is a variational formulation
(associated to the following homogenized system):
\begin{equation} \label{4.6}
-\nabla_y \cdot  [d_1 (\nabla_x u_1(t,x)+\nabla_y u_1^1(t,x,y))]=0 \, \, \, \qquad
\; \; \; \text{ in} \, \,
[0,T] \times \Omega \times Y^{\ast},
\end{equation}

\begin{equation} \label{4.7}
[\nabla_x u_1(t,x)+\nabla_y u_1^1(t,x,y)] \cdot n=0  \, \, \, \qquad 
\; \; \; \; \; \; \; \; \;  \; \; \text{ on} \, \, 
[0,T] \times \Omega \times \Gamma,
\end{equation}

\begin{equation} \label{4.8} 
\begin{split}
&\theta \, \frac{\displaystyle \partial u_1}{\displaystyle \partial t}(t,x)
-\nabla_x \cdot \bigg[ d_1 \, \displaystyle \int_{Y^{\ast}}
(\nabla_x u_1(t,x)+\nabla_y u_1^1(t,x,y)) dy \bigg] \\
&+\theta \, u_1(t,x)
\sum_{j=1}^{\infty} a_{1,j} \, u_j(t,x) 
=d_1 \, \displaystyle \int_{\Gamma} \psi(t,x,y) \, d\sigma(y) \\
&+\theta \, \sum_{j=1}^{\infty} B_{1+j} \, \beta_{1+j,1} \, u_{1+j}(t,x)
\, \, \, 
\; \; \; \text{ in} \, \, [0,T] \times \Omega,
\end{split}
\end{equation}

\begin{equation} \label{4.9}
\bigg[ \displaystyle \int_{Y^{\ast}} (\nabla_x u_1(t,x)+\nabla_y u_1^1(t,x,y)) 
\, dy \bigg] \cdot n =0  \, \, \, \qquad 
\; \; \; \; \; \text{ on} \, \, 
[0,T] \times \partial\Omega,
\end{equation}
where
$$\theta=\int_{Y} \chi(y) dy=\vert Y^{\ast} \vert$$
is the volume fraction of material.
Furthermore, a direct passage to the limit shows that
$$u_1(0,x)=U_1 \, \, \, \qquad \; \; \; \; \; \text{ in} \, \, \Omega.$$

%{\bf{Is it sure? I would have guessed $\theta \,U_1$ }}
%\medskip 

%Taking advantage of the constancy of the diffusion coefficient $d_1$,
Eqs. (\ref{4.6}) and (\ref{4.7}) can be reexpressed as follows:
\begin{equation} \label{4.10}
\triangle_y u_1^1(t,x,y)=0 \, \, \, \qquad \; \; \; \; \; \; \; \; \; \;
\; \; \; \; \; \; \; \; \; \; \; \; \; \; \; \; \; \; \;
\text{ in} \, \,  
[0,T] \times \Omega \times Y^{\ast},
\end{equation}
\begin{equation} \label{4.11}
\nabla_y u_1^1(t,x,y) \cdot n=-\nabla_x u_1(t,x) \cdot n
\, \, \, \qquad \; \; \; \; \;  \text{ on} \, \,
[0,T] \times \Omega \times \Gamma .
\end{equation}
Then, $u_1^1$ satisfying (\ref{4.10})-(\ref{4.11}) can be written as
\begin{equation} \label{4.12}
u_1^1(t,x,y)=\sum_{j=1}^3 \, w_j (y) \, 
\frac{\displaystyle \partial u_1}{\displaystyle \partial x_j}(t,x),
\end{equation}
where $(w_j)_{1 \leq j \leq 3}$
is the family of solutions of the cell problem:
\begin{eqnarray} \label{4.13}
\begin{cases}
-\nabla_y \cdot  [\nabla_y w_j+ \hat{e}_j]=0 \, \, \, \qquad \; \; \; \; \; 
\text{ in} \, \, Y^{\ast} ,\\
(\nabla_y w_j+\hat{e}_j) \cdot n=0 \, \, \, \qquad \; \; \; \; \; \; \; \;
\text{ on} \, \, \Gamma ,\\
y \mapsto w_j(y) \; \; \; \; \; Y-\text{periodic} ,
\end{cases}
\end{eqnarray}
and $\hat{e}_j$ is the $j$-th unit vector of the canonical basis of $\R^3$.
\medskip

By using the relation (\ref{4.12}) in Eqs. (\ref{4.8}) and (\ref{4.9}),
the system (\ref{4.1}) can be immediately derived (cf. \cite{3}).
\medskip

We now consider $ i \geq 2 $, and multiply the first equation of (\ref{1.2}) 
by the same test function  $(t,x) \mapsto \phi_{\epsilon}(t,x,\frac{x}{\epsilon})$
as previously (with $\phi_{\epsilon}$ defined by (\ref{phiep})). We get
%$$\phi_{\epsilon} \equiv \phi(t,x)+\epsilon \, \phi_1 \bigg(t,x,
%\frac{x}{\epsilon} \bigg), $$
%where $\phi \in C^1 ([0,T] \times \overline {\Omega})$ and
%$\phi_1 \in C^1 ([0,T] \times \overline {\Omega}; C_{\#}^{\infty}(Y))$.
%Integrating, the divergence theorem yields
\begin{equation} \label{4.14} 
\begin{split}
&\displaystyle \int_0^T \int_{\Omega_{\epsilon}} 
\frac{\displaystyle \partial u_i^{\epsilon}}{\displaystyle \partial t} \,
\phi_{\epsilon} (t,x,\frac{x}{\epsilon}) \, dt \, dx+
d_i \, \displaystyle \int_0^T \int_{\Omega_{\epsilon}} \nabla_x u_i^{\epsilon} 
\cdot  \nabla_x \left[(t,x) \mapsto \phi_{\epsilon} (t,x,\frac{x}{\epsilon}) \right] \, dt \, dx \\ 
&=-\displaystyle \int_0^T \int_{\Omega_{\epsilon}} u_i^{\epsilon}
\sum_{j=1}^{\infty} a_{i,j} \, u_j^{\epsilon} \, \phi_{\epsilon}(t,x,\frac{x}{\epsilon}) \, dt \, dx+
\frac{\displaystyle 1}{\displaystyle 2}
\displaystyle \int_0^T \int_{\Omega_{\epsilon}}
\sum_{j=1}^{i-1} a_{j,{i-j}} \, u_j^{\epsilon} \, u_{i-j}^{\epsilon} \,
\phi_{\epsilon}(t,x,\frac{x}{\epsilon}) \, dt \, dx \\
&+\displaystyle \int_0^T \int_{\Omega_{\epsilon}}
\sum_{j=1}^{\infty} B_{i+j} \, \beta_{i+j,i} \, u_{i+j}^{\epsilon} \,
\phi_{\epsilon}(t,x,\frac{x}{\epsilon}) \, dt \, dx-
\displaystyle \int_0^T \int_{\Omega_{\epsilon}}
B_i \, u_i^{\epsilon} \, \phi_{\epsilon}(t,x,\frac{x}{\epsilon}) \, dt \, dx.
\end{split}
\end{equation}

The passage to the two-scale limit can be done exactly as in the case when
$u_1^{\epsilon}$ was concerned, and leads to
% thanks 
%to Theorem \ref{tB.2}, we get
\begin{equation} \label{4.15} 
\begin{split}
&\displaystyle \int_0^T \int_{\Omega} \int_{Y^{\ast}}
\frac{\displaystyle \partial u_i}{\displaystyle \partial t}(t,x) \, 
\phi(t,x) \, dt \, dx \, dy \\ 
&+d_i \, \displaystyle \int_0^T \int_{\Omega} \int_{Y^{\ast}}
[\nabla_x u_i(t,x)+\nabla_y u_i^1(t,x,y)] \cdot
[\nabla_x \phi(t,x)+\nabla_y \phi_1(t,x,y)] \, dt \, dx \, dy \\ 
&=-\displaystyle \int_0^T \int_{\Omega} \int_{Y^{\ast}} u_i(t,x)
\sum_{j=1}^{\infty} a_{i,j} \, u_j(t,x) \, \phi(t,x) \, dt \, dx \, dy \\
&+\frac{\displaystyle 1}{\displaystyle 2}
\displaystyle \int_0^T \int_{\Omega} \int_{Y^{\ast}} \sum_{j=1}^{i-1} 
a_{j,{i-j}} \, u_j(t,x) \, u_{i-j}(t,x) \, \phi(t,x) \, dt \, dx \, dy \\
&+\displaystyle \int_0^T \int_{\Omega} \int_{Y^{\ast}}
\sum_{j=1}^{\infty} B_{i+j} \, \beta_{i+j,i} \, u_{i+j}(t,x) \,
\phi(t,x) \, dt \, dx \, dy \\
&-\displaystyle \int_0^T \int_{\Omega} \int_{Y^{\ast}} B_i \, u_i(t,x) \,
\phi(t,x) \, dt \, dx \, dy.
\end{split}
\end{equation}

%The passage to the limit in the infinite sums can be performed since thanks to 
%the assumptions 
%on $a_{i,j}$, $B_j$ and $\beta_{i,j}$,  and to the duality lemma,
%$$ \int_0^T \int_{\Omega_{\epsilon}}
%\sum_{j=K}^{\infty} B_{i+j} \, \beta_{i+j,i} \, u_{i+j}^{\epsilon} \,
%\phi_{\epsilon} \, dt \, dx $$
%$$ \le \int_0^T \int_{\Omega_{\epsilon}}
%\sum_{j=K}^{\infty} \gamma_i\, a_{i, i+j}\, \, u_{i+j}^{\epsilon} \,\, dt \, dx \, ||\phi_{\epsilon}||_{\infty} $$
%$$ \le C_{T,i} \, K^{-\zeta} . $$

An integration by parts shows that (\ref{4.15}) is a variational formulation
associated to the following homogenized system:
\begin{equation} \label{4.16}
- \nabla_y \cdot [d_i (\nabla_x u_i(t,x)+\nabla_y u_i^1(t,x,y))]=0 \, \, \, \qquad
\; \; \; \text{ in} \, \,
[0,T] \times \Omega \times Y^{\ast},
\end{equation}

\begin{equation} \label{4.17}
[\nabla_x u_i(t,x)+\nabla_y u_i^1(t,x,y)] \cdot n=0  \, \, \, \qquad 
\; \; \; \; \; \; \; \; \;  \; \; \; \; \; \; \text{ on} \, \, 
[0,T] \times \Omega \times \Gamma,
\end{equation}

\begin{equation} \label{4.18} 
\begin{split}
&\theta \, \frac{\displaystyle \partial u_i}{\displaystyle \partial t}(t,x)
- \nabla_x \cdot \bigg[ d_i \, \displaystyle \int_{Y^{\ast}}
(\nabla_x u_i(t,x)+\nabla_y u_i^1(t,x,y)) dy \bigg] \\
&=-\theta \, u_i(t,x)
\sum_{j=1}^{\infty} a_{i,j} \, u_j(t,x)+ 
\frac{\displaystyle \theta}{\displaystyle 2}
\sum_{j=1}^{i-1} a_{j,{i-j}}  u_j(t,x) \, u_{i-j}(t,x) \\
&+\theta \, \sum_{j=1}^{\infty} B_{i+j} \, \beta_{i+j,i} \, u_{i+j}(t,x)-
\theta \, B_i \, u_i(t,x) 
\; \; \; \; \; \text{ in} \, \, [0,T] \times \Omega,
\end{split}
\end{equation}

\begin{equation} \label{4.19}
\bigg[ \displaystyle \int_{Y^{\ast}} (\nabla_x u_i(t,x)+\nabla_y u_i^1(t,x,y)) 
\, dy \bigg] \cdot n =0  \, \, \, \qquad 
\; \; \; \; \; \text{ on} \, \, 
[0,T] \times \partial\Omega,
\end{equation}
where $\theta$
%=\int_{Y} \chi(y) dy=\vert Y^{\ast} \vert$$
still is the volume fraction of material. Once again, a direct passage to 
the limit shows that
%Moreover, by continuity
$$u_i(0,x)=0 \, \, \, \qquad \; \; \; \; \; \text{ in} \, \, \Omega.$$
%Taking advantage of the constancy of the diffusion coefficient $d_i$,
Eqs. (\ref{4.16}) and (\ref{4.17}) can be reexpressed as follows:
\begin{equation} \label{4.20}
\triangle_y u_i^1(t,x,y)=0 \, \, \, \qquad \; \; \; \; \; \; \; \; \; \;
\; \; \; \; \; \; \; \; \; \; \; \; \; \; \; \; \; \; \;
\text{ in} \, \,  
[0,T] \times \Omega \times Y^{\ast},
\end{equation}
\begin{equation} \label{4.21}
\nabla_y u_i^1(t,x,y) \cdot n=-\nabla_x u_i(t,x) \cdot n
\, \, \, \qquad \; \; \; \; \;  \text{ on} \, \,
[0,T] \times \Omega \times \Gamma .
\end{equation}
Then, $u_i^1$ satisfying (\ref{4.20}) - (\ref{4.21}) can be written as
\begin{equation} \label{4.22}
u_i^1(t,x,y)=\sum_{j=1}^3 \, w_j (y) \, 
\frac{\displaystyle \partial u_i}{\displaystyle \partial x_j}(t,x),
\end{equation}
where $(w_j)_{1 \leq j \leq 3}$
is the family of solutions of the cell problem (\ref{4.13}).
%, and $\hat{e}_j$ is the $j$-th vector of the canonical basis of $\R3$.
\medskip

By using the relation (\ref{4.22}) in Eqs. (\ref{4.18}) and (\ref{4.19}),
the system (\ref{4.2}) can be immediately derived (cf. \cite{3}).
\medskip

This concludes the proof of our main Theorem (Thm. \ref{t2.1}).

\appendix
\section{Appendix A} \label{appA}

 We introduce in this Appendix some results related to the theory of perforated domains, proven in previous works. In the three Lemmas stated below, $\Omega_{\epsilon}$ is
 a perforated domain satisfying Assumption $0$.
 
\begin{lemma} \label{lA.0}
There exists a constant $C_1>0$ which does not depend on $\epsilon$, such that
when
% following estimate holds: If 
$v \in \text{Lip}\,(\Omega_{\epsilon})$,
then
\begin{equation} \label{A.0} 
{\Vert v \Vert}^2_{L^2(\Gamma_{\epsilon})} \leq
C_1 \, \bigg[ {\epsilon}^{-1} \displaystyle \int_{\Omega_{\epsilon}}
\vert v \vert^2 \, dx+ \epsilon 
\displaystyle \int_{\Omega_{\epsilon}} \vert \nabla_x v \vert^2
\, dx \bigg].
\end{equation}
%where $C_1$ is a constant which does not depend on $\epsilon$.
\end{lemma}

\begin{proof} The inequality (\ref{A.0}) can be easily obtained from the standard trace
theorem by means of a scaling argument, cf. \cite{3a}.
\end{proof}

\begin{lemma} \label{lA.1}
%Suppose that the domain $\Omega_{\epsilon}$ is such that assumption
%(\ref{security}) is satisfied.
There exists a family of linear continuous extension operators
$$P_{\epsilon}: W^{1,p}(\Omega_{\epsilon}) \rightarrow W^{1,p}(\Omega)$$
and a constant $C >0$ which does not depend on $\epsilon$, such that
$$P_{\epsilon} v=v  \; \; \; \text{in } \Omega_{\epsilon} ,$$
and
\begin{equation} \label{A.1}
\int_{\Omega} \vert P_{\epsilon} v \vert^{p} dx \leq
C \int_{\Omega_{\epsilon}} \vert v \vert^p dx \;,
\end{equation}
\begin{equation} \label{A.2}
\int_{\Omega} \vert \nabla  (P_{\epsilon} v) \vert^{p} dx \leq
C \int_{\Omega_{\epsilon}} \vert \nabla v \vert^p dx ,
\end{equation}
for each $v \in W^{1,p}(\Omega_{\epsilon})$ and for any $p \in (1, +\infty)$.
\end{lemma}
\begin{proof}
For the proof of this Lemma,  see for instance \cite{6a}.
\end{proof}

As a consequence of the existence of those extension operators, one can obtain 
Sobolev inequalities in $W^{1,p}(\Omega_{\epsilon})$ with  constants which do not
depend on $\epsilon$.

\begin{lemma} [Anisotropic Sobolev inequalities in perforated domains]
\label{lA.2}
\par\noindent
\par\noindent
(i) 
%For arbitrary $v \in H^1 (0,T; L^2 (\Omega_{\epsilon})) \cap 
%L^2 (0,T; H^1 (\Omega_{\epsilon}))$ and
For $q_1$ and $r_1$ satisfying the
conditions
\begin{eqnarray} \label{A.3}
\begin{cases}
\frac{\displaystyle 1}{\displaystyle r_1}+
\frac{\displaystyle 3}{\displaystyle 2 q_1}=\frac{\displaystyle 3}
{\displaystyle 4} ,
\\
r_1 \in [2, \infty], \; q_1 \in [2, 6] \;  ,\\
\end{cases}
\end{eqnarray}
the following estimate holds (for $v \in H^1 (0,T; L^2 (\Omega_{\epsilon})) \cap 
L^2 (0,T; H^1 (\Omega_{\epsilon}))$):
\begin{equation} \label{A.4}
\Vert v \Vert_{L^{r_1} (0, T; L^{q_1} (\Omega_{\epsilon}))} \leq
c \, \Vert v \Vert_{Q_{\epsilon} (T)} ,
\end{equation}
where $c>0$ does not depend on $\epsilon$, and (we recall that)
\begin{equation} \label{A.5}
\Vert v \Vert^{2}_{Q_{\epsilon} (T)} := \sup_{0 \leq t \leq T}
\displaystyle \int_{\Omega_{\epsilon}} \vert v(t) \vert^2 \, dx+
\displaystyle \int_{0}^{T} \, dt \displaystyle \int_{\Omega_{\epsilon}}
\vert \nabla v(t) \vert^2 \, dx  ;
\end{equation}

\par\noindent
(ii)
% For arbitrary $v \in H^1 (0,T; L^2 (\Omega_{\epsilon})) \cap 
%L^2 (0,T; H^1 (\Omega_{\epsilon}))$ and
For $q_2$ and $r_2$ satisfying the
conditions
\begin{eqnarray} \label{A.6}
\begin{cases}
\frac{\displaystyle 1}{\displaystyle r_2}+
\frac{\displaystyle 1}{\displaystyle q_2}=\frac{\displaystyle 3}
{\displaystyle 4},
\\
r_2 \in [2, \infty], \; q_2 \in [2, 4] \; , \\
\end{cases}
\end{eqnarray}
the following estimate holds (for $v \in H^1 (0,T; L^2 (\Omega_{\epsilon})) \cap 
L^2 (0,T; H^1 (\Omega_{\epsilon}))$):
\begin{equation} \label{A.7}
\Vert v \Vert_{L^{r_2} (0, T; L^{q_2} (\Gamma_{\epsilon}))} \leq
c \, {\epsilon}^{-\frac{3}{2}+\frac{2}{q_2}}  \, 
\Vert v \Vert_{Q_{\epsilon} (T)} ,
\end{equation}
where $c>0$ does not depend on $\epsilon$.
%where $c$ is a positive constant independent of $\epsilon$ and
%the norm $\Vert v \Vert_{Q_{\epsilon} (T)}$ is defined as in (\ref{A.5}).
\end{lemma}

\begin{proof}
For the proof of this Lemma, see \cite{10}.
\end{proof}

\section{Appendix B} \label{appB}

We present in this Appendix 
%Let us summarize some definitions and 
some results on two-scale convergence
(cf. \cite{3}, \cite{3a}, \cite{16}, \cite{9}, \cite{11}, \cite{14}).
%\medskip
%*** I think that we should use  the setting of $\Omega_{\epsilon}$ defined in 1.1. for Prop. B6/Lemma B7. Laurent ***
Up to Prop. B.5, $\Omega$ is a bounded open set of $\R^3$ with smooth boundary, and $Y= [0,1[^3$. Then, for Prop. B.6 and 
Lemma B.7,  $\Omega_{\epsilon}$ is
 a perforated domain satisfying Assumption $0$.
\medskip

We start with the:

\begin{definition} \label{dB.0}
A sequence of functions $v^{\epsilon}$ in $L^2 ([0,T] \times \Omega)$
 two-scale converges to  $v_0 \in L^2 ([0,T] \times \Omega \times Y)$
if
\begin{equation} \label{B.1}
\lim_{\epsilon \rightarrow 0} \int_0^T \int_{\Omega} v^{\epsilon}(t,x) \,
\phi \bigg( t,x,\frac{x}{\epsilon} \bigg) \,dt \,dx=
\int_0^T \int_{\Omega} \int_{Y} v_0 (t,x,y) \, \phi(t,x,y) \,dt \, dx\, dy,
\end{equation}
for all $\phi \in C^1 ([0,T] \times \overline {\Omega}; C_{\#}^{\infty}(Y))$.
\end{definition}

We recall then the following classical Proposition:

\begin{proposition} \label{tB.1}
If $v^{\epsilon}$ is a bounded sequence in $L^2 ([0,T] \times \Omega)$, then
there exists a function $v_0 := v_0 (t,x,y)$ in $L^2 ([0,T] \times \Omega \times Y)$
such that, up to a subsequence,
$v^{\epsilon}$ two-scale converges to $v_0$.
\end{proposition}

Then, the following Proposition is useful for obtaining the limit of the product of
two two-scale convergent sequences.

\begin{proposition} \label{tB.2}
Let $v^{\epsilon}$ be a sequence of functions in $L^2 ([0,T] \times \Omega)$
which two-scale converges to a limit
$v_0 \in L^2 ([0,T] \times \Omega \times Y)$. Suppose furthermore that

\begin{equation} \label{B.2} 
\lim_{\epsilon \rightarrow 0} \int_0^T \int_{\Omega} \vert v^{\epsilon}(t,x)
\vert^2 \, dt \, dx=
\int_0^T \int_{\Omega} \int_{Y} \vert v_0 (t,x,y) \vert^2 \, dt \, dx \, dy .
\end{equation}
Then, for any sequence $w^{\epsilon}$ in $L^2 ([0,T] \times \Omega)$
that two-scale converges to a limit
$w_0 \in L^2 ([0,T] \times \Omega \times Y)$, we get the limit
\begin{equation} \label{B.3}
\begin{split}
\lim_{\epsilon \rightarrow 0} \int_0^T \int_{\Omega} v^{\epsilon}(t,x) &\,
w^{\epsilon}(t,x) \,
\phi \bigg( t,x,\frac{x}{\epsilon} \bigg) \,dt \,dx \\
&=\int_0^T \int_{\Omega} \int_{Y} v_0 (t,x,y) \, w_0 (t,x,y) \, 
\phi(t,x,y) \,dt \, dx\, dy,
\end{split}
\end{equation}
for all $\phi \in C^1 ([0,T] \times \overline {\Omega}; C_{\#}^{\infty}(Y))$.
\end{proposition}

{\bf{Remark}}: Note that, in the setting of this paper, 
identity \eqref{B.2} can be obtained 
by standard computations, used in problems with perforated domains, 
thanks to the existence  
%One uses the properties
 of the extension operators $P_{\epsilon}$
(stated in Lemma \ref{lA.1}). 
%For instance, using $v^{\epsilon}(t,x) := u_1^{\epsilon}$,  
%we see that $P_{\epsilon} v^{\epsilon}$ converges strongly in 
%$L^2$ towards $v_0 := u_1$. 
%As a consequence, $|P_{\epsilon} v^{\epsilon}|^2$ converges towards $v_0^2$ 
%strongly in $L^1$, and therefore it also 2-scales converges towards 
%the same quantity. 
%Finally, the properties of  $P_{\epsilon}$ enable us to obtain identity 
%\eqref{B.2}.
\medskip

The next Propositions yield a characterization of the two-scale limits of 
gradients of bounded sequences $v^{\epsilon}$. This result is crucial for
applications to homogenization problems.

%We identify $H^1(\Omega)=W^{1,2} (\Omega)$, where the Sobolev space
%$W^{1,p} (\Omega)$ is defined by
%$$W^{1,p} (\Omega)=\bigg \{ v \vert v \in L^p(\Omega), 
%\frac{\partial v}{\partial x_i} \in L^p(\Omega), i=1, \ldots, N \bigg \} , $$
%and we denote by $H^1_{\#}(Y)$ the closure of $C^{\infty}_{\#}(Y)$
%for the $H^1$-norm.

\begin{proposition} \label{tB.3}
Let $v^{\epsilon}$ be a bounded sequence in
$L^2 (0,T; H^1 (\Omega))$ that converges weakly to a limit
$v := v(t,x)$ in $L^2 (0,T; H^1 (\Omega))$.
Then, $v^{\epsilon}$ also two-scale converges to
$v$, and there exists a function $v_1 := v_1 (t,x,y)$ in
$L^2 ([0,T] \times \Omega; H^1_{\#} (Y)/\mathbb {R})$ such that, up to extraction of a
subsequence, $\nabla v^{\epsilon}$ two-scale converges to
$\nabla_x v(t,x)+\nabla_y v_1 (t,x,y)$.
\end{proposition}

\begin{proposition} \label{tB.4}
Let $v^{\epsilon}$ be a bounded
sequence in $L^2 ([0,T] \times \Omega)$, such that $\epsilon \nabla_x v^{\epsilon}$
also is  a bounded
sequence in $L^2 ([0,T] \times \Omega)$.
Then, there exists a function $v_1 := v_1 (t,x,y)$ in
$L^2 ([0,T] \times \Omega; H^1_{\#} (Y)/\mathbb {R})$ such that, up to extraction of a
subsequence, $v^{\epsilon}$ and $\epsilon \nabla v^{\epsilon}$ respectively 
two-scale converge to $v_1$ and $\nabla_y v_1$. 
\end{proposition}

The main result of two-scale convergence can be generalized to the case of
sequences defined in $L^2 ([0,T] \times \Gamma_{\epsilon})$.

\begin{proposition} \label{tB.5}
Let $v^{\epsilon}$ be a sequence in $L^2 ([0,T] \times \Gamma_{\epsilon})$
such that

\begin{equation} \label{B.4}
\epsilon \, \int_0^T \int_{\Gamma_{\epsilon}} \vert v^{\epsilon}(t,x) 
\vert^2 \, dt \, d\sigma_{\epsilon}(x) \leq C,
\end{equation}
where $C$ is a positive constant, independent of $\epsilon$.
There exist a subsequence (still denoted by $\epsilon$) and a two-scale limit
$v_0(t,x,y) \in L^2([0,T] \times \Omega; L^2(\Gamma))$ such that
$v^{\epsilon}(t,x)$ two-scale converges to $v_0(t,x,y)$ in the sense that

\begin{equation} \label{B.5} 
 \lim_{\epsilon \rightarrow 0} \, \epsilon \, \displaystyle \int_0^T 
\int_{\Gamma_{\epsilon}} v^{\epsilon}(t,x) \, \phi \bigg(t,x,
\frac{x}{\epsilon} \bigg) \, dt \,
d\sigma_{\epsilon}(x)  = 
 \displaystyle \int_0^T \int_{\Omega} \int_{\Gamma} v_0(t,x,y) \,
\phi(t,x,y) \, dt \, dx \, d\sigma(y) 
\end{equation}
for any function $\phi \in C^1 ([0, T] \times \overline {\Omega};
C_{\#}^{\infty}(Y))$.
\end{proposition}

The proof of Prop. \ref{tB.5} is very similar to the usual two-scale
convergence theorem \cite{3}. It relies on the following lemma \cite{3a}:

\begin{lemma} \label{lB.6}
Let $B=C [\overline{\Omega}; C_{\#} (Y)]$ be the space of continuous
functions $\phi (x, y)$ on $\overline{\Omega} \times Y$ which are
$Y$-periodic in $y$. Then, $B$ is a separable Banach space which is dense
in $L^2 (\Omega; L^2 (\Gamma))$, and such that any function
$\phi (x, y) \in B$ satisfies

\begin{equation} \label{B.6}
\epsilon \, \displaystyle \int_{\Gamma_{\epsilon}} 
\bigg \vert \phi (x,\frac{x}{\epsilon} ) \bigg \vert^2 \, d\sigma_{\epsilon}(x)
\leq C \, \Vert \phi \Vert^2_B, 
\end{equation}
and

\begin{equation} \label{B.7}
 \lim_{\epsilon \rightarrow 0} \, \epsilon \, \int_{\Gamma_{\epsilon}}
\bigg \vert \phi \bigg(x, \frac{x}{\epsilon} \bigg) \bigg \vert^2 \, 
d\sigma_{\epsilon} (x)
=\displaystyle \int_{\Omega} \int_{\Gamma} \vert \phi (x, y) \vert^2 \, dx \, 
d\sigma(y).
\end{equation}

\end{lemma}

\section*{Acknowledgements}

S.L. is supported by GNFM of INdAM, Italy.
L.D. also acknowledges support from
the French ``ANR blanche'' project Kibord: ANR-13-BS01-0004, and by Universit\'e Sorbonne Paris Cit\'e, in the framework of the ``Investissements d'Avenir'', convention ANR-11-IDEX-0005.
\par
L.D. warmly thanks Harsha Hutridurga for very fruitful discussions during the preparation of this work.

\end{document}